\documentclass[english,a4paper,10pt]{article}
\usepackage[utf8]{inputenc}
\usepackage[T1]{fontenc}
\usepackage[english]{babel}

\usepackage[a4paper, total={6in, 8in}]{geometry}

\usepackage{graphbox}
\usepackage{framed}
\usepackage{amssymb}
\usepackage{amsmath}
\usepackage{amsthm}
\usepackage{scalerel}

\usepackage[disable]{todonotes}
\usepackage{hyperref}
\usepackage[vlined,ruled,noend,linesnumbered]{algorithm2e}
\SetArgSty{text}
\DontPrintSemicolon
\usepackage{bussproofs}
\EnableBpAbbreviations
\usepackage{mathtools}
\usepackage{tikz}
\usetikzlibrary{shapes.geometric, patterns}
\usetikzlibrary{arrows.meta}
\usetikzlibrary{calc}
\usetikzlibrary{shapes.misc, positioning}
\usepackage{csquotes}
\usepackage{booktabs}
\usepackage{thmtools}
\usepackage{thm-restate}
\usepackage{enumitem}

\usepackage{stmaryrd}

 \newtheorem{theorem}{Theorem}
 \newtheorem{lemma}[theorem]{Lemma}
 
 \theoremstyle{definition}
 \newtheorem{definition}[theorem]{Definition}
 \newtheorem{example}{Example}

\newif\ifhideproofs

\newif\iftechnicalReport

\newif\ifplotDataInTikz

\newif\ifappendix

\plotDataInTikztrue 
\hideproofsfalse
\technicalReporttrue 

\ifhideproofs
  \usepackage{environ}
  \NewEnviron{hide}{}

\fi

\newcommand{\alisa}[2][]{\ttodo{#1}{Alisa}{cyan}{#2}}
\newcommand{\stefan}[2][]{\ttodo{#1}{Stefan}{yellow}{#2}}
\newcommand{\ttodo}[4]{\ifthenelse{\equal{#1}{inline}}{\todo[inline, author=#2, color = 
#3]{#4}}{\todo[color=#3]{#2: #4}}}

\newcommand{\hide}[1]{}

\newcommand{\wlg}{w.l.o.g.\ }
\newcommand{\Wlog}{W.l.o.g.\ }
\newcommand{\wrt}{w.r.t.\ }
\newcommand{\st}{s.t.\ }
\newcommand{\ie}{i.e.\ }
\newcommand{\eg}{e.g.\ }
\newcommand{\etc}{etc.\ }
\newcommand{\cf}{cf.\ }

\newcommand{\aka}{a.k.a.\ }

\newlength{\myl}
\newcommand{\longsquigarrow}[1]{
    \settowidth{\myl}{$~_{#1}$}
    \raisebox{-0.01cm}{\xymatrix@C=\myl{
            {}\ar@{~>}[r]^{~_{#1}}&{}
        }
    }
}

\newcommand{\blue}[1]{{#1}}

\newcommand{\ACzero}{\ensuremath{\mathsf{AC^0}}\xspace}
\newcommand{\FO}{\textup{FO}}

\newcommand{\SRIQ}{\mathcal{SRIQ}\xspace}
\newcommand{\HornSRIQ}{\textsl{Horn-}$\SRIQ$\xspace}

\newcommand{\leaf}{\ensuremath{\mathsf{leaf}}\xspace}
\newcommand{\sink}{\ensuremath{\mathsf{sink}}\xspace}

\newcommand{\card}[1]{\lvert #1\rvert}

\newcommand{\PTime}{\textsc{P}\xspace}
\newcommand{\PSpace}{\textsc{PSpace}\xspace}
\newcommand{\NP}{\textsc{NP}\xspace}

\newcommand{\NExpTime}{\textsc{NExpTime}\xspace}

\newcommand{\Amc}{\ensuremath{\mathcal{A}}\xspace}

\newcommand{\Imc}{\ensuremath{\mathcal{I}}\xspace}

\newcommand{\Lmc}{\ensuremath{\mathcal{L}}\xspace}

\newcommand{\Omc}{\ensuremath{\mathcal{O}}\xspace} 
 
\newcommand{\Tmc}{\ensuremath{\mathcal{T}}\xspace}

\newcommand{\EL}{\ensuremath{\mathcal{E}\hspace{-0.1em}\mathcal{L}}\xspace}

\newcommand{\ELI}{\ensuremath{\mathcal{ELI}}\xspace}
\renewcommand{\L}{\ensuremath{\mathcal{L}}\xspace}
\newcommand{\Lcq}{\ensuremath{\L_{\textit{cq}}}\xspace}
\newcommand{\DLLite}{\ensuremath{\textsl{DL-Lite}}\xspace}
\newcommand{\DLLiteR}{\ensuremath{\textsl{DL-Lite}_R}\xspace}

\newcommand{\q}{\ensuremath{\mathbf{q}}\xspace} 
\newcommand{\x}{\ensuremath{\vec{x}}\xspace}
\newcommand{\y}{\ensuremath{\vec{y}}\xspace}
\renewcommand{\a}{\ensuremath{\vec{a}}\xspace}

\newcommand{\tup}[1]{\langle #1 \rangle}

\newcommand{\NC}{\ensuremath{\textsf{N}_\textsf{C}}\xspace}

\newcommand{\sig}{\ensuremath{\textsf{sig}}\xspace}

\newcommand{\answer}{\ensuremath{\vec{a}}\xspace}

\newcommand{\ind}{\ensuremath{\text{\textsf{ind}}}\xspace}
\newcommand{\ex}[1]{\ensuremath{\textsf{\upshape{#1}}}\xspace}
 
\newcommand{\p}{\ensuremath{\mathcal{P}}\xspace}

\newcommand{\R}{\ensuremath{\mathfrak{D}}\xspace}

\newcommand{\Rx}{\ensuremath{\R_\mathsf{x}}\xspace}
\newcommand{\Rcq}{\ensuremath{\R_\mathsf{cq}}\xspace}
\newcommand{\Rsk}{\ensuremath{\R_\mathsf{sk}}\xspace}
\newcommand{\Rtcq}{\ensuremath{\R_\mathsf{tcq}}\xspace}
\newcommand{\Rtsk}{\ensuremath{\R_\mathsf{tsk}}\xspace}

\newcommand{\OP}{\ensuremath{\mathsf{OP}}\xspace}
\newcommand{\OPx}{\ensuremath{\OP_\mathsf{x}}\xspace}
\newcommand{\OPtx}{\ensuremath{\OP_\mathsf{tx}}\xspace}
\newcommand{\sk}{\ensuremath{_\mathsf{sk}}\xspace}
\newcommand{\cq}{\ensuremath{_\mathsf{cq}}\xspace}

\newcommand{\el}{\ensuremath{\ell}\xspace}

\newcommand{\ds}{\ensuremath{\mathcal{D}}\xspace}

\newcommand{\MOPO}{\textbf{(MP)}\xspace}
\newcommand{\CONJ}{\textbf{(C)}\xspace}
\newcommand{\TAUT}{\textbf{(T)}\xspace}
\newcommand{\EXISTS}{\textbf{(E)}\xspace}
\newcommand{\gMOPO}{\textbf{($\mathbf{MP_s}$)}\xspace}
\newcommand{\gCONJ}{\textbf{($\mathbf{C_s}$)}\xspace}
\newcommand{\gEXISTS}{\textbf{($\mathbf{E_s}$)}\xspace}
\newcommand{\gEXISTSb}{\textbf{($\mathbf{E_s'}$)}\xspace}
\newcommand{\tMOPO}{\textbf{(TMP)}\xspace}
\newcommand{\COAL}{\textbf{(COAL)}\xspace}
\newcommand{\SEP}{\textbf{(SEP)}\xspace}
\newcommand{\LOR}{\textbf{(DISJ)}\xspace}
\newcommand{\tw}{\ensuremath{\mathsf{t}}\xspace}

\newcommand{\m}{\ensuremath{\mathfrak{m}}\xspace}
\newcommand{\mtree}{\ensuremath{\m_{\mathsf{t}}}\xspace}
\newcommand{\msize}{\ensuremath{\m_{\mathsf{s}}}\xspace}

\newcommand{\lnext}{\scaleobj{0.7}{\bigcirc}}

\newcommand{\luntil}{\,\mathcal{U}}

\newcommand{\lsince}{\,\mathcal{S}}
\newcommand{\tem}{\ensuremath{\text{\textsf{tem}}}\xspace}
\newcommand{\Imf}{{\ensuremath{\mathfrak{I}}}\xspace}
\newcommand{\Zbb}{{\ensuremath{\mathbb{Z}}}\xspace}

\begin{document}



\title{Finding Good Proofs for Answers to Conjunctive Queries Mediated by
Lightweight Ontologies (Technical Report)}

\author{Christian Alrabbaa \and Stefan Borgwardt \and Patrick Koopmann \and Alisa Kovtunova}

\date{}
\maketitle

\begin{abstract}
In ontology-mediated query answering, access to incomplete data sources is
mediated by a conceptual layer constituted by an ontology. To correctly compute
answers to queries, it is necessary to perform complex reasoning over the
constraints expressed by the ontology. In the literature, there exists a
multitude of techniques incorporating the ontological knowledge into queries.
However, few of these approaches were designed for comprehensibility of the
query answers. In this article, we try to bridge these two qualities by adapting
a proof framework originally applied to axiom entailment for conjunctive query
answering. We investigate the data and combined complexity of determining the
existence of a proof below a given quality threshold, which can be measured in
different ways. By distinguishing various parameters such as the shape of a
query, we obtain an overview of the complexity of this problem for the
lightweight ontology languages \DLLiteR and \EL, and also have a brief
look at temporal query answering.
\todo{optional: remove \EL if we need space}
\end{abstract}

\section{Introduction}


Explaining description logic (DL) reasoning has a long tradition, starting with the first works on \emph{proofs} for standard DL entailments~\cite{DeMc-96,DBLP:conf/ecai/BorgidaFH00}.
A popular and very effective method \blue{is} \emph{justifications}, which simply point out the axioms 
from an ontology that are responsible for an 
entailment~\cite{ScCo03,DBLP:conf/ki/BaaderPS07,Pena-09,Horr-11}.
More recently, work has resumed on techniques to find proofs for explaining more complex logical consequences~\cite{DBLP:conf/semweb/HorridgePS10,KaKS-DL17,LPAR23:Finding_Small_Proofs_for,ABB+-DL20,ABB+-CADE21}.
On the other hand, if a desired entailment does not hold, one needs different explanation 
techniques such as
abduction~\cite{DBLP:conf/ijcai/Koopmann21,EX_RULES_ABDUCTION,%
DBLP:conf/kr/CalvaneseOSS12} or
counterinterpretations~\cite{DBLP:conf/ki/AlrabbaaHT21}.
%
Explaining answers to conjunctive queries (CQs) has also been investigated before, in the form of abduction for missing answers over \DLLite ontologies~\cite{DBLP:conf/kr/CalvaneseOSS12}, provenance for positive answers in \DLLite and \EL~\cite{DBLP:conf/ijcai/CalvaneseLOP019,DBLP:conf/ijcai/BourgauxOPP20}, as well as proofs for \DLLite query answering~\cite{DBLP:conf/otm/BorgidaCR08,Stefanoni-11,DBLP:conf/ekaw/CroceL18}.

Here, we also investigate proofs for CQ answers, inspired
by~\cite{DBLP:conf/otm/BorgidaCR08,Stefanoni-11,DBLP:conf/ekaw/CroceL18}, but
additionally consider the problem of generating \emph{good} proofs according to
some quality measures and provide a range of complexity results mostly focussing on
\DLLiteR. 
In addition to classical OMQA, we also have a brief look at
explaining inferences over temporal data using a query language incorporating
metric temporal operators.
Our results are based on a framework developed for proofs of standard DL
reasoning~\cite{LPAR23:Finding_Small_Proofs_for}.
There, proofs are formalized as directed, acyclic hypergraphs and proof quality can be measured in different ways.
We mainly consider the \emph{size} (the number of formulas) of a proof as well as its \emph{tree size}, which corresponds to the size when the proof is presented in a tree-shaped way (which may require repeating subproofs), as it is often done in practice~\cite{KaKS-DL17,DBLP:conf/dlog/AlrabbaaBDFK20}.
The quest for good proofs is formalized as a search problem in a so-called \emph{derivation 
structures} produced by a \emph{deriver}, which specifies the possible
inferences.

In this paper, we consider two different kinds of derivers for
generating proofs for CQ answers. These loosely correspond to the approaches
in~\cite{DBLP:conf/otm/BorgidaCR08,Stefanoni-11,DBLP:conf/ekaw/CroceL18}, but
are generalized to apply to a larger class of DLs. Specifically, our structures
rely on a translation of DLs to \emph{existential
rules}~\cite{DBLP:journals/ws/CaliGL12}, and thus apply to any DL that can be
expressed in this formalism.
%
%
\blue{One} deriver, denoted by~\Rcq and inspired
by~\cite{Stefanoni-11,DBLP:conf/ekaw/CroceL18}, focuses on the derivation of
CQs, which can be derived from other CQs and ontology axioms.
Inferences in~\Rcq are logically sound, but can be harder to understand.
The reason is the local scope of existential quantification in a CQ, which
forces atoms connected by the same variables to be carried along inferences
they are not relevant for. This problem is circumvented with the deriver~\Rsk,
which relies on a Skolemized version of the TBox. This allows one to focus
on inferences of single atoms that are only later aggregated into the final CQ,
leading to simpler sentences within the proof.
%
Focusing on the particular case of \DLLiteR, 
we consider the complexity of the decision problems of finding proofs of (tree) size below a given threshold~$n$ in these derivation structures.
We find that for \DLLiteR and any DL in which CQ answering is UCQ-rewritable, all of these problems (regardless of derivation structure and quality measure) are in \ACzero in data complexity.
In combined complexity, these problems are \NP-complete in general, but polynomial when considering only acyclic queries and tree size.
%
%
To explain answers to \emph{temporal} queries, we extend our derivers with new
inference schemes to deal with metric temporal operators,  allowing us to lift
some of our results also to this setting.
\stefan{TODO: cite extended version on arxiv!}


\section{Preliminaries}\label{sec:preliminaries}


\paragraph*{Proofs}

In our setting, a
\emph{logic}~$\Lmc=(\mathcal{S}_\Lmc,\models_\Lmc)$ consists of a set
$\mathcal{S}_\Lmc$ of
\emph{$\Lmc$-sentences} and a \emph{consequence relation}
\mbox{${\models_\Lmc}\subseteq P(\mathcal{S}_\Lmc)\times \mathcal{S}_\Lmc$} between
\emph{\Lmc-theories} (subsets of
\Lmc-sentences) and
single \Lmc-sentences; we usually write only $\models$ instead of $\models_\Lmc$.
We assume that the \emph{size}~$|\eta|$ of an \L-sentence~$\eta$ is defined in
some way, \eg by the number of symbols in~$\eta$.
We require that \Lmc is \emph{monotonic}, \ie that $\Tmc\models\eta$ implies $\Tmc'\models\eta$ for all $\Tmc'\supseteq\Tmc$.
%
%
For example, \L could be \emph{first-order logic} or some DL.

As in~\cite{LPAR23:Finding_Small_Proofs_for,ABB+-DL20,ABB+-CADE21}, we view
proofs as directed hypergraphs (see the appendix for details).

\begin{definition}[Derivation Structure]\label{def:derivation-structure}
  A \emph{derivation structure} $\ds = (V, E, \el)$ over a theory~$\mathcal{U}$ is a directed, labeled hypergraph that is
  \begin{itemize}
  \item \emph{grounded}, \ie
  every leaf $v$ in~\ds is labeled by $\el(v)\in\mathcal{U}$; and
  \item \emph{sound}, \ie for every hyperedge $(S,d)\in E$, the entailment
$\{\el(s)\mid s\in S\}\models\el(d)$ holds.
\end{itemize}
\end{definition}

We call hyperedges $(S,d)\in E$ \emph{inferences} or \emph{inference steps},
with $S$ being the \emph{premises} and $d$ the \emph{conclusion}, and may write them like 
\begin{center}
 \AXC{$p$}
 \AXC{$p\to q$}
 \BIC{$q$}
 \DP
 \quad
 or
 \quad
 \scalebox{0.8}{
 \begin{tikzpicture}[
  scale=0.8,
  baseline=(x3.north),
  block/.style={
  draw,
  rounded rectangle,
  minimum width={width("$A \sqsubseteq \exists r.(\exists r.A)$")-5pt}, 
  },
  tbox/.style={
  draw=gray, text=gray,
  rounded rectangle,
  minimum width={width("$A \sqsubseteq \exists r.(\exists r.A)$")-5pt},
  },
  he/.style={rounded corners=5pt,->}]
  \path (-1,3.5) node[block] (x1) {$p$}
       (5,3.5) node[block] (x2) {$p\to q$}
       (2,2.75) node[block] (x3) {$q$};
  \draw[he] (x1) -- ($(x1)!0.5!(x2)$)-- (x3);
  \draw[he] (x2) -- ($(x1)!0.5!(x2)$)-- (x3);
  \end{tikzpicture}
  }
\end{center}
\emph{Proofs} are special derivation structures that derive a goal sentence.

\begin{definition}[Proof]\label{def:proof}
  Given a sentence~$\eta$ and a theory~$\mathcal{U}$, a \emph{proof of $\mathcal{U}\models\eta$} is a \blue{finite}
  derivation structure $\p = (V, E,\el)$ over~$\mathcal{U}$ such that
  \begin{itemize}
    \item\label{item1:def-proof-non-redundancy} \p contains exactly one
        sink~$v_\eta\in V$, which is labeled by~$\eta$,
    \item \p is acyclic, and
    \item every vertex has at most one incoming hyperedge, \ie there exist no
two hyperedges $(S_1,v)$, $(S_2,v)\in E$ with $S_1\neq S_2$.
  \end{itemize}
A \emph{tree proof} is a proof that is a tree.
A \emph{subproof} $S$ of a hypergraph~$H$ is a subgraph of~$H$ that is a proof 
with $\leaf(S)\subseteq\leaf(H)$.
\end{definition}

To compute proofs, we assume that there is some reasoning system or calculus that defines
a derivation structure for a given entailment~$\eta$, and \blue{the structure} may contain several
proofs for that entailment. Formally, a \emph{deriver} $\R$ for a
logic $\Lmc$ takes as input an $\Lmc$-theory $\mathcal{U}$ and an
$\Lmc$-sentence~$\eta$, and returns a \blue{(possibly infinite)} derivation structure
$\R(\mathcal{U},\eta)$ over $\mathcal{U}$ that describes all inference
steps \blue{that} $\R$ could perform in order to derive $\eta$ from $\mathcal{U}$.
This derivation structure is not necessarily computed explicitly, but can
be accessed through an oracle (which checks, for example, whether an inference conforms to the underlying calculus).
The task of finding a
good proof then corresponds to finding a \blue{(finite)} proof that can be homomorphically
mapped into this derivation structure and which is minimal according to some measure of proof 
quality.
We consider two such measures here: the \emph{size} of a proof $\p=(V,E,\ell)$ is $\msize(\p):=|V|$,\footnote{Since every vertex has at most one incoming hyperedge, the size of $E$ is at most quadratic in~$|V|$.} 
and the \emph{tree size} $\mtree(\p)$
is the size of a tree unraveling of~\p~\cite{ABB+-CADE21}.
The \emph{depth} of~\p is the length of the longest path from a leaf to the sink (see appendix).

\paragraph*{DLs and Existential Rules}

We assume that the reader is familiar with DLs, in particular \DLLiteR~\cite{JAR-2007}, 
where theories $\mathcal{U}=\Tmc\cup\Amc$ are called \emph{ontologies} or \emph{knowledge bases} and are composed of a TBox~\Tmc and an ABox~\Amc.
Many DL ontologies can be equivalently expressed using the formalism of existential
rules~\cite{DBLP:journals/ws/CaliGL12}. Existential rules are first-order sentences of the form
$\forall\y,\vec{z}.\,\psi(\y,\vec{z})\to\exists\vec{u}.\,\chi(\vec{z},\vec{u})$, with the
\emph{body} $\psi(\y,\vec{z})$ and the \emph{head} $\chi(\vec{z},\vec{u})$ being conjunctions of atoms of the form $A(x)$ or $P(x_1, x_2)$, for a concept name $A$, role name $P$ and
terms $x$, $x_1$ and $x_2$, which are individual names or variables from
$\vec{z}$, $\vec{u}$ and~$\y$. We usually omit the universal
quantification.
Notable DLs
that can be equivalently expressed as sets of existential rules are \EL,
\HornSRIQ and \DLLiteR.

\paragraph*{Conjunctive Queries}

In this paper, we want to construct proofs for ontology-mediated conjunctive query entailments.
A \emph{conjunctive query (CQ)} $\q(\x)$ is an expression of the
form $\exists \y.\,\phi(\x,\y)$, where $\phi(\x,\y)$ is a conjunction of atoms. The variables in $\x = (x_1,\dots, x_n)$ are called \emph{answer
variables} and the variables in $\y$ \emph{existentially quantified variables}.
If \blue{$n=0$}, then $\q(\x)$ is called \emph{Boolean}.
ABox assertions are a special case of Boolean CQs with only one atom
an no variables.
%
A tuple $\answer$ of individual names from $\Amc$ is a \emph{certain answer} to
$\q(\x)$ over $\Tmc\cup\Amc$, in symbols $\Tmc\cup\Amc\models\q(\answer)$, if,
for any model of $\Tmc\cup\Amc$, the sentence $\q(\answer)$ is true in this
model.
%
Any CQ $\q(\x)=\exists \y.\,\phi(\x,\y)$ is associated with the set of atoms in 
$\phi$, so we will write, \eg $A(z) \in \q$.

\begin{example}\label{ex}
For the following $\DLLiteR$ ontology and query, we have
$\Tmc\cup\{B(\ex{b})\}\models \q(\ex{b})$.
\begin{align*}
\Tmc &= \{A \sqsubseteq \exists R, \qquad  \exists R^- \sqsubseteq \exists T, \qquad
B \sqsubseteq \exists P, \qquad \exists P^- \sqsubseteq \exists S, \qquad 
P \sqsubseteq R^- \} \\
\q(y'') &= \exists x, x', x'', y, y', z, z'.\, R(x, y) \,\land\, T (y, z) \,\land\, T (y', z) \,\land\, R(x', y') \,\land\,  S(x', z') \\ &\hspace*{8.5cm} \,\land\, S(x'', z') \,\land\, P(y'',x'').
\end{align*}
In the next section, we explore different ways to explain this inference (see Figures~\ref{fig:ex-existential} and~\ref{fig:ex-skolem}).
\end{example}

\section{Derivation Structures for Certain Answers}\label{sec:DS4OMQA}

In the following, let $\Tmc\cup\Amc$ be a knowledge base in some DL \Lmc, $\q$ a conjunctive query, and $\vec{a}$ a certain answer, \ie $\Tmc\cup\Amc\models\q(\vec{a})$, which we want to explain.
We can use derivation structures over \Lcq (the extension of \Lmc with all Boolean CQs) to explain query answers.
For example, the following derivation step involving the ontology from Example~\ref{ex} is a sound inference:
\begin{center}
\AXC{$B(\ex{b})$}
\AXC{$\Tmc$}
\BIC{$\q(\ex{b})$}
\DP
\end{center}

However, to define a derivation structure that yields proofs suitable for
explanations to users, inferences that only make small deduction steps are more valuable. 
For this purpose, we define derivers that
capture which inference steps are admitted.
For TBox entailment,
in~\cite{LPAR23:Finding_Small_Proofs_for,ABB+-DL20,ABB+-CADE21}, we considered
derivers based on the inference schemas used by a
consequence-based reasoner. To obtain proofs for CQ entailment, we follow the
ideas of \emph{chase} procedures that replace atoms in CQs by other atoms by
\enquote{applying} rules to
them~\cite{DBLP:journals/tcs/FaginKMP05,DBLP:journals/ws/CaliGL12,JAR-2007,DBLP:conf/otm/BorgidaCR08}.
%
We will introduce two derivers that represent different paradigms of what constitutes a proof.

\subsection{The CQ Deriver}
\blue{Similarly to the approach used in~\cite{Stefanoni-11,DBLP:conf/ekaw/CroceL18}, inferences in our first deriver, $\Rcq$, always produce Boolean CQs.}
This deriver is defined for DLs that can be expressed using existential
rules.
%
%
%
An inference step is obtained by matching the left-hand side of a rule to part of a CQ and then replacing it by the right-hand side.
%
%
%
For example, starting from $\exists z.\,{P}(\ex{b},z)$ and
${P}(x,y)\to{R}(y,x)$, we can apply the substitution $\{x\mapsto\ex{b}, y
\mapsto z\}$ to obtain $\exists z.\,R(z,\ex{b})$.
Additionally, we allow to keep any of the replaced atoms from the original CQ,
\eg to produce the conclusion $\exists z.\,{P}(\ex{b},z)\land R(z,\ex{b})$.
A second type of inference allows one to combine two Boolean CQs using
conjunction.
To duplicate variables, we
additionally introduce tautological rules such as $P(x,z)
\to \exists z'.\, P(x,z')$, which yields $\exists z,z'.\,{P}(\ex{b},z)\land
{P}(\ex{b},z')$ when combined with $\exists z.\,{P}(\ex{b},z)$.
%
Finally, we use an inference schema that allows us to replace constants by variables, \eg to capture that $\exists z.\,P(\ex{b},z)$ implies $\exists x,z.\, P(x,z)$.

\begin{figure}
  \begin{framed}
    \begin{minipage}{0.45\textwidth}
      \centering
      \AXC{$\exists\vec{x}.\,\phi(\vec{x})$}
      \AXC{$\psi(\vec{y},\vec{z})\to\exists\vec{u}.\,\chi(\vec{z},\vec{u})$}
      \RL{\MOPO}
      \BIC{$\exists\vec{w}.\rho(\vec{w})$}
      \DP
    \end{minipage}
    \hfill
    \begin{minipage}{0.45\textwidth}
      \centering
      \AXC{$\exists\vec{x}.\,\phi(\vec{x})$}
      \AXC{$\exists\vec{y}.\,\psi(\vec{y})$}
      \RL{\CONJ}
      \BIC{$\exists\vec{x},\vec{y}'.\phi(\vec{x})\land\psi(\vec{y}')$}
      \DP
    \end{minipage}

    \bigskip

    \begin{minipage}{0.45\textwidth}
      \centering
      \AXC{}
      \RL{\TAUT}
      \UIC{$\phi(\vec{x},\vec{y})\to\exists\vec{x}.\,\phi(\vec{x},\vec{y})$}
      \DP
    \end{minipage}
    \hfill
    \begin{minipage}{0.45\textwidth}
      \centering
      \AXC{$\exists\vec{x}.\,\phi(\vec{x},\vec{a})$}
      \RL{\EXISTS}
      \UIC{$\exists\vec{x},\vec{y}.\,\phi(\vec{x},\vec{y})$}
      \DP  
    \end{minipage}
  \end{framed}
  \vspace*{-\baselineskip}
  \caption{Inference schemas for \Rcq. \blue{\MOPO and \TAUT refer to \emph{modus ponens} and \emph{tautology}.}}
  \label{fig:dcq}
\end{figure}

The detailed inference schemas can be found in Figure~\ref{fig:dcq}. \MOPO is admissible only if there exists a substitution $\pi$ such
that $\pi(\psi(\vec{y},\vec{z}))\subseteq\phi(\vec{x})$, and then
$\rho(\vec{w})$ is the result of replacing \emph{any subset of}
$\pi(\psi(\vec{y},\vec{z}))$ in $\phi(\vec{x})$ by \emph{any subset of}
$\pi(\chi(\vec{z},\vec{u}'))$, where the variables $\vec{u}$ are renamed into
new existentially quantified variables $\vec{u}'$ to ensure that they are
disjoint with $\vec{x}$.
In \CONJ, we again rename the variables $\vec{y}$ to $\vec{y}'$ to avoid overlap
with $\vec{x}$. Since every ABox assertion corresponds to a ground CQ,
this inference also allows one to collect ABox assertions into a single CQ.
\TAUT \blue{introduces an existential rule that} allows us, together with~\MOPO, to create copies of variables in CQs (see Fig.~\ref{fig:ex-existential}).
Finally, \EXISTS transforms individual names \blue{in some positions} into existentially quantified variables.

\begin{definition}[CQ Deriver]\label{def:cq-ds}
  $\Rcq(\Tmc\cup\Amc,\q(\vec{a}))$ is a derivation structure over $\Tmc\cup\Amc$
  with vertices labeled by the axioms in $\Tmc\cup\Amc$ and all Boolean CQs over the signature of $\Tmc\cup\Amc$, 
  and its hyperedges represent all \blue{possible} instances of~\MOPO, \CONJ, \TAUT, and~\EXISTS over these vertices.
%
An \emph{(admissible) proof in $\Rcq(\Tmc\cup\Amc,\q(\vec{a}))$} is a proof of $\Tmc\cup\Amc\models\q(\vec{a})$ that has a homomorphism into this derivation structure.
\end{definition}
%
It is easy to check that the inferences used by \Rcq are sound. Moreover, we
can show that \blue{they are} complete, \ie that any CQ
entailed by $\Amc\cup\Tmc$ has a proof in $\Rcq(\Tmc\cup\Amc,\q(\vec{a}))$
(see Lemma~\ref{lem:transformation}).
%
A proof for Example~\ref{ex} \wrt \Rcq is depicted in Figure~\ref{fig:ex-existential}.

\begin{figure}
\centering\begin{footnotesize}
\begin{tikzpicture}[scale=0.8,
block/.style={
draw,
rounded rectangle,
minimum width={width("$A \sqsubseteq \exists r.(\exists r.A \sqcap B)$")-5pt}, minimum height=1.8em,
},
tbox/.style={
draw=gray, text=gray,
rounded rectangle,
minimum width={width("$A \sqsubseteq \exists r.(\exists r.A \sqcap B)$")-5pt},
minimum height=1.8em,
},
he/.style={rounded corners=10pt,->}]
\path (-1,1.8) node[block] (x1) {$B(\ex{b})$}
		(5,1.8) node[tbox] (x2) {$B\sqsubseteq\exists P$}
		(2,0.8) node[block] (x3) {$\exists x''.\, P(\ex{b},x'')$}
		(8,0.8) node[tbox] (y) {$\exists P^-\sqsubseteq \exists S$}
		(5,-0.3) node[block] (y3) {$\exists x'',z'\,
P(\ex{b},x'')\,\land\,S(x'',z')$}
		(-3,-0.3) node[tbox] (y6) {$P\sqsubseteq R^-$}
		(1,-1.4) node[block] (y5) {$\exists x'',z'.\,
R(x'',\ex{b})\,\land\,S(x'',z')\,\land\,P(\ex{b},x'')$}
		(9,-1.4) node[tbox] (y8){$\exists R^-\sqsubseteq \exists T$}
        (6.5,-2.5) node[block] (y7) {$\exists x'',z,z'.\,
R(x'',\ex{b})\,\land\,T(\ex{b},z)\,\land\,S(x'',z')\,\land\,P(\ex{b},x'')$}
         (-2.1,-3.2) node[tbox, dashed] (z0) {$R(x,y) \,\land\, T(y,z) \to \exists
y'.\, R(x,y') \,\land\, T(y',z) $} 
	   (5.8,-4.2) node[block] (y9) {$\exists x'',y',z,z'.\,
R(x'',\ex{b})\,\land\,T(\ex{b},z)\,\land\,T(y',z)\,\land\,
R(x'',y')\,\land\dots $}
	   	(-2,-4.8) node[tbox,dashed] (y4)  {$S(x,z)\to S(x,z) $}
	    (6.5,-5.7) node[block] (z1) { \dots }
	     (-1.7,-6.2) node[tbox,dashed] (z11){$R(x'',y')\,\land\,S(x'',z') \to \exists
x'.\, R(x',y')\,\land\,S(x',z')$}
		    (6.5,-7.6) node[block] (z2) {\dots}
		     (-1.5,-7.6) node[tbox,dashed] (z21){$R(x'',y) \to\exists x.\, R(x,y)$}
	     (2.5,-8.7) node[block] (z3) {$\exists x,x',x'',y',z,z'.\,
R(x,\ex{b})\,\land\,T(\ex{b},z)\,\land\,T(y',z)\,\land\,R(x',y')\,\land\,S(x',z'
)\,\land\,S(x'',z')\,\land\,P(\ex{b},x'')$}
	  (2.5,-9.9) node[block,,fill=gray!20] (z) { $\exists x,x', x'', y, y',z, z'.\, R(x, y)
\,\land\, T (y, z) \,\land\, T (y', z) \,\land\, R(x', y') \,\land\,  S(x', z')
\,\land\, S(x'', z') \,\land\, P(\ex{b},x'')$};
\draw[he, draw=gray,text=gray]  ($(z0)+(0,0.7)$) --  (z0) node[yshift=11pt,xshift=6pt] {\tiny{\TAUT}};	
\draw[he, draw=gray,text=gray]  ($(y4)+(0,0.7)$) --  (y4) node[yshift=11pt,xshift=6pt] {\tiny{\TAUT}};	
\draw[he, draw=gray,text=gray]  ($(z11)+(0,0.7)$) --  (z11) node[yshift=11pt,xshift=6pt] {\tiny{\TAUT}};		
\draw[he, draw=gray,text=gray]  ($(z21)+(0,0.7)$) --  (z21) node[yshift=11pt,xshift=6pt] {\tiny{\TAUT}};		
\draw[he] (x1) -- ($(x1)!0.5!(x2)$) node[yshift=3pt] {\tiny{\MOPO}} -- (x3);
\draw[he] (x2) -- ($(x1)!0.5!(x2)$) -- (x3);
\draw[he] (y) -- ($(y)!0.5!(x3)$) node[yshift=3pt] {\tiny{\MOPO}} -- (y3);
\draw[he] (x3) -- ($(y)!0.5!(x3)$) -- (y3);
\draw[he] (y3) -- ($(y3)!0.5!(y6)$) node[yshift=3pt] {\tiny{\MOPO}}  -- (y5);
\draw[he] (y6) -- ($(y3)!0.5!(y6)$) -- (y5);
\draw[he] (y5) -- ($(y5)!0.5!(y8)+(1.5,0)$) node[yshift=3pt] {\tiny{\MOPO}}  -- (y7);
\draw[he] (y8) -- ($(y5)!0.5!(y8)+(1.5,0)$) -- (y7);
\node(y9') at ($(y9)+(0.7,0.22)$){};
\draw[he] (y7) -- (y9');
\draw[he] (z0) -- ($(y7)!0.5!(y9')$) node[xshift=7pt,yshift=-4pt] {\tiny{\MOPO}}  -- (y9');
\draw[he] ($(y9')-(0,0.65)$) -- (z1);
\draw[he] (y4) -- ($(y9')!0.5!(z1)$) node[xshift=7pt,yshift=-4pt] {\tiny{\MOPO}}  -- (z1);
\draw[he] (z1) -- (z2);
\draw[he] (z11) -- ($(z1)!0.35!(z2)$) node[xshift=7pt,yshift=-4pt] {\tiny{\MOPO}} -- (z2);
\draw[he] (z2) -- ($(z2)!0.5!(z21)$) node[yshift=3pt] {\tiny{\MOPO}} -- (z3);
\draw[he] (z21) -- ($(z2)!0.5!(z21)$)-- (z3);
\draw[he] (z3) -- (z) node [xshift=6pt,yshift=11pt] {\tiny{\EXISTS}};
\end{tikzpicture}
\end{footnotesize}
\caption{A CQ proof for Example~\ref{ex} (inferences \EXISTS and \TAUT are delayed to
the last steps)}
\label{fig:ex-existential}
\end{figure}
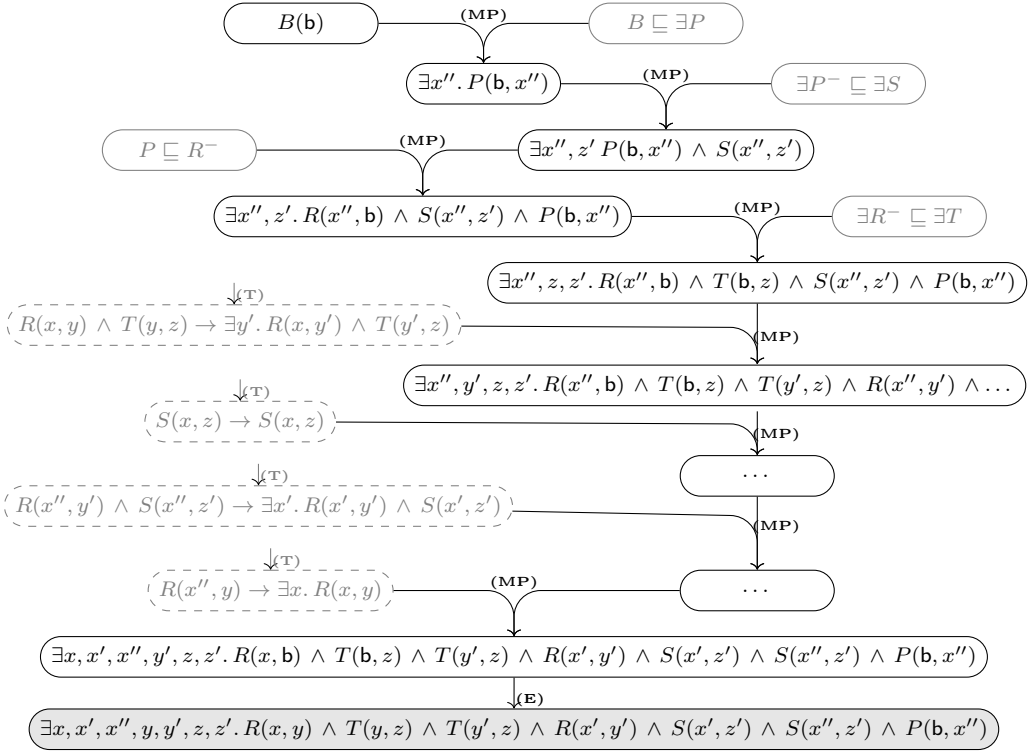

\subsection{Skolemized Derivation Structure}

To explain a Boolean CQ, using a derivation structure
that works on CQs seems natural. However, a downside is that
we have to \enquote{collect} quantified variables along the proof
\blue{and label vertices with complex expressions.}
\blue{Since the inference rules apply on sub-expressions, it may be challenging to understand on which part of the CQ an inference is performed.}
The
problem is that we cannot separate inference steps on the same variable without
affecting soundness, as the existential quantification only applies locally in
the current CQ.
To follow our example: $x''$ and~$z'$ in
Figure~\ref{fig:ex-existential} are connected to each other and to the constant
$\ex{b}$, and thus have to be kept together: although $\exists x'',z'.P(\ex{b},x'')\land
S(x'',z')$ implies $\exists x''.P(\ex{b},x'')$ and $\exists x'',z'.S(x'',z')$, those two
CQs do not imply the original CQ anymore.
\blue{Another problem is that verifying the correctness of an \MOPO-step actually requires solving an \NP-hard problem from the user, \ie finding a matching substitution for a conjunction of atoms (unless we augment the proof by showing the substitutions as well).}
To overcome these issues, we
consider a second type of deriver that relies on Skolemization, and is inspired
by the approach from~\cite{DBLP:conf/otm/BorgidaCR08}.

%
%
This deriver, \Rsk, mainly operates on ground CQs, and requires the theory to
be \emph{Skolemized}. This means that it cannot contain existential quantification,
it may however contain function symbols. To Skolemize existential rules, for each
existentially quantified variable a fresh function symbol is introduced; for the
CI 
$\exists P^-\sqsubseteq \exists S$  this results in $P(x,y) \to S(y,g(y))$, where
$g$ is a unary function symbol whose argument denotes the dependency on
the variable~$y$ shared between the body and head of the rule.
Let $\Tmc^s$ be the set of Skolemized rules resulting from this transformation \blue{and note that the entailments $\Tmc\cup\Amc\models\q(\a)$ and $\Tmc^s\cup\Amc\models\q(\a)$ are equivalent for CQs $\q(\x)$ that do not use function symbols}.
Our deriver internally considers two kinds of formulas: 1) CQs that may
use function symbols and 2) rules of the form
$\forall\vec{x}.\phi(\vec{x})\rightarrow\psi(\vec{x})$, where $\psi(\vec{x})$
may now contain function terms, but no further quantified variables.
%
Since we are only interested in CQs that are entailed by~$\Tmc^s\cup\Amc$, we can assume \wlg that this entailment can be shown solely using domain elements denoted by ground terms, \eg $f(f(\ex{a}))$, which allows us to eliminate variables from most of the inferences. 
%
For example, instead of 
$\exists x'',z'.\, P(\ex{b},x'')\land S(x'',z')$ in Figure~\ref{fig:ex-existential}
we now use $P(\ex{b},f(\ex{b}))\land S(f(\ex{b}),g(f(\ex{b})))$. 
Since these atoms do not share variables, in our derivation structure we mainly need to consider inferences on single atoms, which allows for more fine-grained proofs (see Figure~\ref{fig:ex-skolem}).
Only at the end we need to compose atoms to obtain a CQ.

\begin{figure}
  \begin{framed}
    \centering
    \def\defaultHypSeparation{\hskip .35em}
    \AXC{$\alpha_1(\vec{t}_1)$}
    \AXC{$\dots$}
    \AXC{$\alpha_n(\vec{t}_n)$}
    \AXC{$\psi(\vec{y},\vec{z})\to\,\chi(\vec{z})$}
    \RL{\gMOPO}
    \QIC{$\beta(\vec{s})$}
    \DP
    \\
    \AXC{$\alpha_1(\vec{t}_1)$}
    \AXC{$\dots$}
    \AXC{$\alpha_n(\vec{t}_n)$}
    \RL{\gCONJ}
    \TIC{$\alpha_1(\vec{t}_1)\land\dots\land\alpha_n(\vec{t}_n)$}
    \DP
    \qquad
    \AXC{$\phi(\vec{t})$}
    \RL{\gEXISTS}
    \UIC{$\exists\vec{x}.\phi(\vec{x})$}
    \DP
  \end{framed}
  \vspace*{-\baselineskip}
  \caption{Inference schemas for \Rsk.}
  \label{fig:dsk}
\end{figure}
The simplified inference schemas are shown in Figure~\ref{fig:dsk}.
In~\gMOPO, $\alpha_i(\vec{t}_i)$ and $\beta(\vec{s})$ are ground atoms with terms
composed from individual names and Skolem functions, and likewise~$\chi(\vec{z})$ may 
contain Skolem functions; similar to \MOPO, we require that
there is a substitution~$\pi$ such that
$\pi(\psi(\vec{y},\vec{z}))=\{\alpha_1(\vec{t}_1),\dots,\alpha_n(\vec{t}_n)\}$
and $\beta(\vec{s})\in\pi(\chi(\vec{z}))$.
In~\gEXISTS, $\vec{t}$ is now a vector of ground terms which may contain function symbols.
Since $\gMOPO$ works only with ground atoms, \gCONJ and \gEXISTS can now
only be used at the end of a proof to obtain the desired CQ (see Figure~\ref{fig:ex-skolem}).
%
Moreover, we do not need a version of~\TAUT here since it would be trivial for ground atoms.
Its effects in \Rcq can be simulated here due to the fact that the same atom can be used several times as a premise for \gMOPO or \gCONJ.
%

\begin{definition}[Skolemized Deriver]\label{def:cq-ds-skolem}
  The derivation structure $\Rsk(\Tmc^s\cup\Amc,\q(\vec{a}))$ is defined similarly to Definition~\ref{def:cq-ds}, but using $\Tmc^s$ and the inference schemas~\gMOPO,~\gCONJ and~\gEXISTS.
\end{definition}

Though different presentations with different advantages and disadvantages, it
is not hard to translate proofs based on \Rsk into proofs in \Rcq and vice
versa.

\begin{restatable}{lemma}{LemTransformation}\label{lem:transformation}
Any proof~$\p$ in $\Rcq(\Tmc\cup\Amc,\q(\vec{a}))$ can be transformed into
a proof in $\Rsk(\Tmc^s\cup\Amc,\q(\vec{a}))$ of size polynomial in the
sizes of~$\p$ and~\Tmc, and conversely any proof
$\p$ in $\Rsk(\Tmc^s\cup\Amc,\q(\vec{a}))$ can be transformed into
a proof in $\Rcq(\Tmc\cup\Amc,\q(\vec{a}))$ of size polynomial in the
sizes of~$\p$ and~\Tmc. \blue{The latter also holds for tree proofs.}
\end{restatable}
\alisa{Rev2: "can the translation be performed in polynomial time?"}
\stefan{I would say yes. Should we state this explicitly?}
However, it is not the case that \emph{minimal} proofs are equivalent for these two derivers, \ie a CQ could have only large proofs in one of them but small proofs in the other, or vice versa.

This lemma also shows that our derivation structures are complete, \ie if $\Tmc\cup\Amc\models\q(\a)$ holds, then we can provide a proof for it.
To see this, consider the minimal Herbrand model~$H$ of $\Tmc^s\cup\Amc$, which can be computed using the \emph{(Skolem) chase} procedure for existential rules---essentially, applying the rules step-by-step to obtain new ground atoms, in a way very similar to~\gMOPO.
This model is a universal model for CQ answering over $\Tmc\cup\Amc$, which means that $\Tmc\cup\Amc\models\q(\a)$ implies $H\models\q(\a)$, which, in turn, means that there must be a proof in $\Rsk(\Tmc^s\cup\Amc,\q(\a))$, and hence by Lemma~\ref{lem:transformation} also one in $\Rcq(\Tmc\cup\Amc,\q(\a))$.
For convenience, we assume in the following that TBoxes are silently Skolemized
when constructing derivation structures using $\Rsk$, that is, we identify
$\Rsk(\Tmc\cup\Amc,\q(\vec{a}))$ with $\Rsk(\Tmc^s\cup\Amc,\q(\vec{a}))$.


\begin{figure}
\centering\begin{footnotesize}
\begin{tikzpicture}[scale=0.8,
block/.style={
draw,
rounded rectangle,
},
tbox/.style={
draw=gray, text=gray,
rounded rectangle,
},
he/.style={rounded corners=10pt,->}]
\path  (3,2) node[block] (x1) {$B(\ex{b})$}
		(9,2) node[tbox] (x2) {$B \sqsubseteq \exists P$} 
		(6,1) node[block] (x3) {$P(\ex{b},f(\ex{b}))$}
		(12,1) node[tbox] (y0) {$\exists P^-\sqsubseteq \exists S$} 
		(0,1) node[tbox] (y1) {$P\sqsubseteq R^-$} 
		(9,0) node[block] (y2) {$S(f(\ex{b}),g(f(\ex{b})))$}
		(3,0) node[block] (y5) {$R(f(\ex{b}),\ex{b})$}
		(-3.2,0) node[tbox] (y6){$\exists R^-\sqsubseteq \exists T$} 
        (-0.1,-1) node[block] (y7) {$T(\ex{b},h(\ex{b}))$}
      	%
	     %
	     (4.5,-3) node[block] (z3) { $R(f(\ex{b}),\ex{b}) \land T(\ex{b},h(\ex{b})) \land T(\ex{b},h(\ex{b})) \land R(f(\ex{b}),\ex{b}) \land S(f(\ex{b}),g(f(\ex{b})))\land S(f(\ex{b}),g(f(\ex{b}))) \land P(\ex{b},f(\ex{b}))$}
	  (4.5,-4.2) node[block,fill=gray!20] (z) { $\exists x,x', x'', y, y',z, z'.\, R(x, y) \,\land\, T (y, z) \,\land\, T (y', z) \,\land\, R(x', y') \,\land\,  S(x', z') \,\land\, S(x'', z') \,\land\, P(\ex{b},x'')$};
\draw[he] (x1) -- ($(x1)!0.5!(x2)$) node[yshift=3pt] {\tiny{\gMOPO}} -- (x3);
\draw[he] (x2) -- ($(x1)!0.5!(x2)$) -- (x3);
%
\draw[he] (y0) -- ($(y0)!0.5!(x3)$) node[yshift=3pt] {\tiny{\gMOPO}} -- (y2);
\draw[he] (x3) -- ($(y0)!0.5!(x3)$) -- (y2);
%
\draw[he] (y1) -- ($(y1)!0.5!(x3)$) node[yshift=3pt] {\tiny{\gMOPO}} -- (y5);
\draw[he] (x3) -- ($(y1)!0.5!(x3)$) -- (y5);
\draw[he] (y5) -- ($(y5)!0.5!(y6)$) node[yshift=3pt] {\tiny{\gMOPO}} -- (y7);
\draw[he] (y6) -- ($(y5)!0.5!(y6)$) -- (y7);
\draw[he] (x3) -- ($(z3)+(0,1)$) node[xshift=8pt,yshift=-8pt] {\tiny{\gCONJ}}  --(z3);
\draw[he] (y2)-- ($(z3)+(0,1)$) --(z3);
\draw[he] ($(y2)+(0.3,-0.33)$)-- ($(z3)+(0,1)$) --(z3);
\draw[he] ($(y5)-(0.3,0.34)$)-- ($(z3)+(0,1)$) --(z3);
\draw[he] ($(y7)-(0.3,0.33)$)-- ($(z3)+(0,1)$) --(z3);
\draw[he] (y5)-- ($(z3)+(0,1)$) --(z3);
\draw[he] (y7)-- ($(z3)+(0,1)$) --(z3);
\draw[he] (z3) --(z) node [xshift=8pt,yshift=14pt] {\tiny{\gEXISTS}};
\end{tikzpicture}
\end{footnotesize}
\caption{A Skolemized proof for Example~\ref{ex}}\label{fig:ex-skolem}
\end{figure}
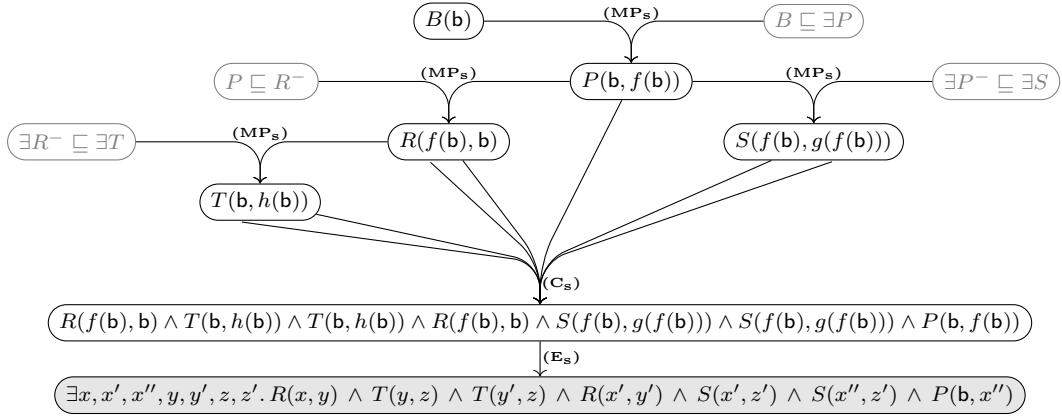

\hide{
\section{Example}
Recall what is a tree witness~\cite{DBLP:conf/rweb/KontchakovZ14}. Let $\Omc$ be a TBox and an RBox in normal form and $\q(x)$ a CQ. Consider a pair $\tw = (\tw_r, \tw_i)$ of disjoint sets of terms in $\q(x)$, where $\tw_i$ (interior) is non-empty and contains only existentially quantified variables and $\tw_r$ (roots), on the other hand, can be empty and can contain answer variables and individual names. Set
$\q_\tw=\{S(\mathbf{z}) \in \q \; | \; \mathbf{z}\subseteq \tw_i \cup \tw_r \; \text{ and }\mathbf{z} \not\subseteq \tw_r\}$
We say that $\tw$ is a tree witness for $\q(x)$ and $\Tmc$ generated by $\exists R$ if the following two conditions are satisfied:
\begin{itemize}
\item there exists a homomorphism $h$ from $\q_\tw$ to the canonical model $C^{\exists R}_\Tmc(a)$, for some $a$, such that $\tw_r = \{z | h(z) = a \}$ and $\tw_i$ contains the remaining variables in $\q_\tw$,
\item $\q_\tw$ is a minimal subset of $\q$ such that, for any $y \in \tw_i$, every atom in $\q$ containing $y$ belongs to $\q_\tw$.
\end{itemize}
Tree witnesses $\tw$ and $\tw'$ are said to be conflicting if $\q_\tw \cap \q_{\tw'} \neq \emptyset $ (in other words, the interior of one tree witness contains a root or an interior variable of the other, which makes it impossible to have both tree witnesses mapped into the anonymous part of $C_{\Omc,\Amc}$ at the same time). A set of tree witnesses is said to be independent if any two distinct tree witnesses in the set are non-conflicting. 

Consider a more complex example.

\begin{example}[\cf Example 35 in~\cite{DBLP:conf/rweb/KontchakovZ14}]
There are four tree witnesses for $\q(x, y'')$ and $\Omc$:
\begin{itemize}
\item $\tw^1=(\tw^1_r, \tw^1_i)$ generated by $\exists T$ with $\tw^1_r = \{y, y'\}$ and $\tw^1_i = \{z\}$ and $\q_\tw^1 = \{ T (y, z), T (y', z) \}$;
\item $\tw^2 = (\tw^2_r, \tw^2_i )$ generated by $\exists S$ with $\tw^2_r = \{x', x''\}$ and $\tw^2_i = \{z'\}$ and $\q_\tw^2 = \{ S(x', z'), S(x'', z') \}$;
\item $\tw^3 = (t^3_r , t^3_i )$ generated by $\exists R$ with $\tw^3_r = \{x, x'\}$ and $\tw^3_i = \{y, y', z\}$ and \\
$\q_\tw^3 = \{ R(x, y), T (y, z), T (y', z), R(x', y') \}$;
\item $\tw^4 = (t^4_r , t^4_i )$ generated by $\exists P$ with $\tw^4_r = \{y, y',y''\}$ and $\tw^4_i = \{x, x', z,z',x''\}$ and \\
$\q_\tw^4 = \{ R(x, y), T (y, z), T (y', z), R(x', y'), S(x',z'), S(x'',z'), P(y'',x'') \}$.

\end{itemize}
Clearly, $\tw^4$ is conflicting with the rest; $\tw^3$ is conflicting with $\tw^1$ but not with $\tw^2$ despite sharing a common root. Thus, we have the following independent sets of tree witnesses: $\emptyset$, $\{\tw^1\}$, $\{\tw^2\}$, $\{\tw^3\}$, $\{\tw^4\}$, $\{\tw^1,\tw^2\}$, $\{\tw^2,\tw^3\}$.

Let us assume we know that which independent set correspond to a positively evaluated conjunct in the tree witness rewriting. For example, $\{\tw^4\}$. On its basis, we can build a proof over the ABox $\{B(\ex{b})\}$ as following. 

We have $\tw^4_r$ mapped to an ABox individual and  $\tw^4_i$ mapped either to Skolem terms (see Figure~\ref{fig:ex-skolem}) or to existential quantified variables (Figure~\ref{fig:ex-existential}).

\end{example} 
}

\section{The Complexity of Finding Good Proofs}\label{sec:results}

\blue{It is our intution that proofs in \Rsk are more comprehensible than using \Rcq because they are more modular. Additionally, we assume that \emph{small} proofs (\wrt size \msize or tree size \mtree) are more comprehensible than larger ones (but one can certainly also consider other measures~\cite{ABB+-DL20,ABB+-CADE21}). Therefore, we now study the \emph{complexity} of finding small proofs automatically (which is independent of the comprehensibility of the resulting proofs).}
More precisely, we are interested in the following decision problem \blue{$\OPx(\Lmc,\m)$ for a deriver $\Rx\in\{\Rcq,\Rsk\}$, a DL $\Lmc$ (\eg $\DLLiteR$), and a measure $\m\in\{\msize,\mtree\}$}:
given an \Lmc-KB $\Tmc\cup\Amc$, a query~$\q(\vec{x})$ with certain answer~$\vec{a}$,
and a natural number~$n$ \blue{(in binary encoding)}, is there a proof~\p for~$\q(\vec{a})$ in
$\Rx(\Tmc\cup\Amc,\q(\vec{a}))$ with $\m(\p)\le n$?
%
%
To better distinguish the complexity of finding small proofs from that of
query answering, we assume $\Tmc\cup\Amc\models\q(\a)$ as \blue{prerequisite},
which fits the intuition that users request an explanation only after
they know that $\a$ is a certain answer.
Lemma~7 in~\cite[]{ABB+-CADE21} shows that, instead of looking for arbitrary
proofs and
homomorphisms into the derivation structure, one can restrict the search to
subproofs of $\Rx(\Tmc\cup\Amc,\q(\vec{a}))$, which we will often do
implicitly.

It is common in the context of OMQA to distinguish between \emph{data complexity}, where only 
the data varies, and
\emph{combined complexity}, where also the influence of the other inputs is taken into account.
This raises the question whether the bound~$n$ is
seen as part of the \blue{input} or not. It turns out that fixing~$n$
trivializes the data complexity,
\blue{which can be shown by observing that there are only constantly many ABoxes (up to isomorphism) that can entail a given CQ.}

\begin{restatable}{theorem}{TheACzero}\label{th:TheACzero}
  \blue{For a constant bound~$n$, $\OPx(\Lmc,\m)$ is in \ACzero in data complexity.}
\end{restatable}
One may argue that, since the size of the proof depends on \Amc, the bound
$n$ on the proof size should be considered part of the data as well. Under this
assumption, our decision problem is not necessarily in \ACzero anymore. For
example, consider the \EL TBox $\{ \exists r.A\sqsubseteq A \}$ and
$q(x)\leftarrow A(x)$. For every $n$, there is an ABox
$\Amc$ such that $A(a)$ is entailed by a seqeuence of~$n$ role assertions, and thus needs a proof of
size at least~$n$. Deciding whether this query admits a bounded proof is thus
as hard as deciding whether it admits an answer at all in $\Amc$, \ie
$\PTime$-hard~\cite{DBLP:conf/dlog/Rosati07}. However, we at
least stay in $\ACzero$ for DLs over which CQs are rewritable, \eg \DLLiteR~\cite{JAR-2007}, \blue{because the number of (non-isomorphic) proofs that we need to consider is bounded by the size of the rewriting, which is constant in data complexity.}

\begin{restatable}{theorem}{TheACzeroRewritable}\label{th:TheACzeroRewritable}
  \blue{If all CQs are UCQ-rewritable over \Lmc-TBoxes, then $\OPx(\Lmc,\m)$ is in \ACzero in data complexity.}
\end{restatable}

We now consider the combined complexity. In~\cite{LPAR23:Finding_Small_Proofs_for,ABB+-CADE21}, we
established general upper bounds for finding proofs of bounded size.
These results depend only on 
the size
of the derivation structure obtained 
for the given input.
Both $\Rcq$
and $\Rsk$ may produce derivation structures of \blue{infinite} size, as $\Rcq$ contains CQs
of arbitrary size, and $\Rsk$ also has Skolem
terms of arbitrary nesting depth.
However, we can sometimes bound the number
of relevant Skolem terms in~\Rsk by considering only
the part of the minimal Herbrand model~$H$ that is necessary to satisfy the query~$\q(\a)$.
For example, in logics with the \emph{polynomial witness 
property}~\cite{DBLP:journals/ai/GottlobKKPSZ14}, including \DLLiteR, we know that any 
query that is entailed is already satisfied after polynomially many chase steps used to 
construct~$H$.
In particular, this means that the nesting depth of Skolem terms in a proof is bounded polynomially (in the size of the TBox and the query), and hence the part of $\Rsk(\Tmc^s\cup\Amc,\q(\a))$ that we need to search for a (small) proof is bounded exponentially.
%
%
%
For such structures, our results from~\cite{LPAR23:Finding_Small_Proofs_for,ABB+-CADE21} give us
a \NExpTime-upper bound for size, and a \PSpace-upper bound for tree size, upon
which we can improve with the following lemma.

\begin{restatable}{lemma}{LemProofSize}\label{lem:proof-size}
There is a polynomial~$p$ such that for any \DLLiteR KB $\Tmc\cup\Amc$,
CQ $\q(\vec{x})$, and certain answer~$\vec{a}$, there is a proof
in $\Rsk(\Tmc\cup\Amc,\q(\vec{a}))$
of tree size at most $p(\lvert\Tmc\rvert,\lvert
\q(\vec{x})\rvert)$.
\end{restatable}

A direct consequence of Lemmas~\ref{lem:transformation}
and~\ref{lem:proof-size} is the upper bound in the following theorem.
\blue{The lower bound can be shown by a reduction from query answering over \DLLiteR ontologies; since we assume $\Tmc\cup\Amc\models\q(\a)$ as a prerequisite for our decision probem, the only difficulty here is to construct an artificial proof for this entailment that ensures that $\Tmc'\cup\Amc'\models\q(\a)$ holds (for some extension $\Tmc'\cup\Amc'$ of $\Tmc\cup\Amc$), but such that the artificial proof is larger than any \enquote{natural} proof using only $\Tmc\cup\Amc$.}
\stefan{too detailed?}

\begin{restatable}{theorem}{TheNPDLLite}\label{th:TheNPDLLite}
  $\OPx(\DLLiteR,\m)$ is \NP-complete.
\end{restatable}


%

To obtain tractability, we can restrict the shape of the
query. Recall that the \emph{Gaifman graph} of a query $\q$ is the undirected
graph
using the terms of $\q$ as nodes and has an edge between terms
occurring together in an atom.
%
A query is \emph{tree-shaped} if its Gaifman graph is a tree.
\begin{restatable}{theorem}{LemPTimeCombined}\label{th:ptime-combined}
 Given a \DLLiteR KB $\Tmc\cup\Amc$ and a tree-shaped CQ $\q(\x)$ with
certain answer $\vec{a}$, one can compute in polynomial time a proof of minimal
tree size in $\Rsk(\Tmc\cup\Amc,\q(\vec{a}))$.
\end{restatable}

\blue{The central property used in the proof of Theorem~\ref{th:ptime-combined} is that for tree size}
%
every atom in $\q(\vec{a})$ has a
separate proof, even if two atoms are proven in the same way.
%
To avoid this redundancy,
one could think about modifying $\gEXISTS$ slightly:
\begin{center}
  \AXC{$\phi(\vec{t})$}
  \RL{\gEXISTSb}
  \UIC{$\exists\vec{x}.\phi'(\vec{x})$}
  \DP, provided there exists $\sigma\colon\vec{x}\rightarrow\vec{t}$
s.t.\ $\phi'(\vec{x})\sigma = \phi(\vec{t})$
\end{center}
Denote the
resulting deriver
by $\Rsk'$.
Using $\gEXISTSb$, we can 
derive $\exists x,y.\, A(x)\wedge A(y)$
from $A(a)$
; with $\gEXISTS$, the premise would need to be
$A(a)\wedge A(a)$.
\blue{However}, this
modification is already sufficient to
make our problem \NP-hard for tree-shaped queries, and even \emph{without a TBox}. Because
$\Rcq$ allows
duplication of atoms
using~\TAUT, the same problem arises there.
\blue{Similarly, $\Rsk$ allows to use the same atom multiple times in one application of~\gEXISTS, but \msize only counts them once.}
\begin{restatable}{theorem}{TheNPHardModified}\label{the:np-hard-trees}
  \blue{For tree-shaped CQs, $\OPx'(\Lmc,\mtree)$ is \NP-hard. The same holds for $\OP\sk(\Lmc,\msize)$ and $\OP\cq(\Lmc,\mtree)$.}
\end{restatable}

\hide{

Possible results:
\begin{itemize}
  \item For (hypergraph) size, this problem is in NP. Is it also NP-hard for DL-Lite (assuming the answer~$a$, the mapping~$\pi$, and the derivation structure to be given, perhaps via oracle functions)?
  
  For data complexity, we would have to construct an abstract version of the derivation structure (replacing all constants by placeholders\footnote{These are not really variables since they only stand in for arbitrary ABox individuals that we don't know yet.}) to first search for abstract proofs there and then try to instantiate them by checking which of the leaves can be mapped to the ABox? This would essentially result in a UCQ to be evaluated over the ABox, \ie should be in \ACzero(?) \alisa{here we may problems with defining $\pi$ over placeholders. Alternatively, if we apply the procedure as above: Skolemization, grounding, matching with the CQ, it feels like in $\PTime$ data complexity}

  \item For tree size, we previously had an NP-hardness result for exponential derivation structure with unary encoding (\aka polynomial threshold), but only for \ELI. On the other hand, it is not clear how to manipulate the above (exponential) structure to \eg use the Dijkstra algorithm to obtain a better upper bound.
  
  For data complexity, it should be \ACzero due to similar arguments as above.

  \item For depth, the same problems appear, \ie we need an NP lower bound for \DLLite.
  
  \item Weighted sum/product behave similarly to size and could be used to formulate something about probabilistic databases. The problem is again that we only have lower bounds for \EL/$\EL_\bot$ so far (see the other main.tex in this repo).
\end{itemize}
}

\section{Metric Temporal CQs}

We now consider the problem of finding a small proof for \emph{temporal} query answering, generalizing the results from Sections~\ref{sec:DS4OMQA} and~\ref{sec:results}.
In this setting, TBox axioms hold globally, \ie at all time points, the ABox contains information about the state of the world in different time intervals, and the query contains (metric) temporal operators.

An \emph{interval} $\iota$ is a \textit{nonempty} subset of $\Zbb$ of the form $[t_1,t_2]$, where $t_1,t_2 \in \Zbb\cup\{\infty\}$ and $t_1 \leq t_2$ (for simplicity, we write $[\infty,t_2]$ for $(-\infty,t_2]$ and $[t_1,\infty]$ instead of $[t_1,\infty)$);\footnote{\blue{This allows us to avoid considering special cases in the interval arithmetic below.}} 
 %
$t_1$ and $t_2$ are encoded in binary. 
A temporal ABox~$\Amc$ is a finite set of facts of the form $A(a)@\iota$ or $P(a,b)@\iota$, where $A(a)$ and $P(a,b)$ are assertions and $\iota$ is an interval. The fact $A(a)@\iota$ states that $A(a)$ holds throughout the interval $\iota$. We denote by $\tem(\Amc)$ the \blue{multi}{}set of intervals 
that occur in $\Amc$ and $\lvert \tem(\Amc) \rvert$ is the sum of their lengths.
A \emph{temporal interpretation} $\Imf = ( \Delta^{\Imf}, (\Imc_i)_{i\in \Zbb})$, is a collection of DL interpretations $\Imc_i = (\Delta^\Imf, \cdot^{\Imc_i})$, $i \in \Zbb$, over~$\Delta^\Imf$.
%
$\Imf$ \emph{satisfies}
  a TBox axiom~$\alpha$ if each $\Imc_i$, $i\in\Zbb$, satisfies~$\alpha$,
  and it satisfies a temporal assertion $\alpha@\iota$ if each $\Imc_i$, $i\in\iota$, satisfies~$\alpha$.
%

We use the finite-range positive version of metric temporal conjunctive queries
(MTCQs) introduced
in~\cite{DBLP:conf/ruleml/BorgwardtFK19,DBLP:journals/tplp/BorgwardtFK22},
combining CQs with MTL operators~\cite{AlHe-JACM94,GuJO-ECAI16,BBK+-FroCoS17}.

\begin{definition}\label{def:mtcq}
An MTCQ is of the form $\q(\x,w)=\phi(\x)@w$, where $\phi$ is built according to
\[
  \phi::= \psi \mid \top \mid 
  \phi \land \phi \mid \phi \lor \phi 
  \mid \boxminus_{I}\phi \mid \boxplus_{I}\phi
\mid \phi \luntil_{I} \phi \mid \phi \lsince_{I} \phi,
\]
with $w$ an interval variable, $\psi$ a CQ, $I$ a finite interval with non-negative
endpoints, and $\x$ the free variables of all CQs in~$\phi$. A certain answer
to $\q(\x,w)$ over $\Tmc\cup\Amc$ is a pair $(\answer,\iota)$ such that
$\answer\subseteq \ind(\Amc)$, $\iota$ is an interval and, for any $t\in \iota$
and any model $\Imf$ of $\Tmc\cup \Amc$, 
we have~$\Imf,t \models \phi(\answer)$ according to Table~\ref{fig:tcq-semantics}. We denote this as $\Tmc\cup\Amc\models
\q(\answer,\iota)$.

\end{definition}
%
%
\begin{table}
 \caption{Semantics of (Boolean) MTCQs for
$\Imf=(\Delta^\Imf,(\Imc_i)_{i\in\Zbb})$ and $i\in\Zbb$.}
 \label{fig:tcq-semantics}
  \setlength{\tabcolsep}{10pt}
\begin{tabular}{lp{9cm}}
  \toprule
  $\phi$ & $\Imf, i \models \phi$ iff \\
  \midrule
  CQ~$\psi$ & $\Imc_i \models \psi$ \\
  $\top$ & true \\
  $\phi \land \psi$ & $\Imf, i \models \phi \text{ and } \Imf, i \models \psi$ \\

  $\phi \lor \psi$ & $\Imf, i \models \phi \text{ or } \Imf, i \models \psi$ \\
    $\boxplus_{I}\phi$ & $\forall k \in I \text{ such that } \Imf, i+k \models \phi$ \\
      $\boxminus_{I}\phi$ & $\forall k \in I \text{ such that } \Imf, i-k \models \phi$ \\
  $\phi\luntil_{I} \psi$ & $\exists k \in I$ such that $\Imf, i+k \models \psi$ and $\forall j: 0 \leq j < k: \Imf, i+j \models \phi$ \\
  $\phi\lsince_{I} \psi$ & $\exists k\in I$ such that $\Imf, i-k \models \psi$ and $\forall j: 0 \leq j < k: \Imf, i-j \models \phi$ \\
  \bottomrule
\end{tabular}
\end{table}


For temporal extensions of Definitions~\ref{def:cq-ds} and~\ref{def:cq-ds-skolem}, we will interpret $A \sqsubseteq A'$ now as the global temporal rule $A(x) \to A'(x)$ holding in any possible interval.
\begin{center}
  \AXC{$(\exists\vec{x}.\,\phi(\vec{x}))@\iota$}
  \AXC{$\psi(\vec{y},\vec{z})\to\exists\vec{u}.\,\chi(\vec{z},\vec{u})$}
  \RL{\tMOPO}
  \BIC{$(\exists\vec{w}.\rho(\vec{w}))@\iota$}
  \DP
\end{center}

Similarly, we need temporal versions of~\CONJ and \EXISTS, where all CQs are annotated with the same interval variable.
In addition, we need an inference for disjunctive MTCQS:
\begin{center}
  \AXC{$\phi(\vec{x})@\iota$}
  \RL{\LOR}
  \UIC{$(\phi(\vec{x})\lor \psi(\vec{y}))@\iota $}
  \DP
\end{center}

To provide a proof for a \emph{temporal} query, we need to be able to coalesce, \ie merge intervals:
\begin{center}
  \AXC{$\exists \x_1.\,\phi(\x_1)@\iota_1$}
  \AXC{$\dots$}
  \AXC{$\exists \x_n.\,\phi(\x_n)@\iota_n$}
  \RL{\COAL}
  \TIC{$(\exists \x.\,\phi(\x))@\bigcup_{i=1}^n\iota_i$}
  \DP
\end{center}
where $\bigcup_{i=1}^s\iota_i$ is a single interval and $\phi(\x_1),\dots,\phi(\x_n)$ are identical up to variable renaming. 
On the other hand, we also need an inverse operation to shrink intervals:
\begin{center}
  \AXC{$\exists\vec{x}.\,\phi(\x)@\iota$}
  \RL{\SEP}
  \UIC{$\exists\vec{x}.\,\phi(\x)@\iota'$}
  \DP
\end{center}
where $\iota' \subseteq \iota$.
%
%
Both inferences are needed to infer all intervals~$\iota$ with $\Tmc\cup\Amc\models\exists\vec{x}.\,\phi(\x)@\iota$.


Finally, we need inferences for the temporal operators, where for $\luntil_{[r_1,r_2]}$ we only consider the case where $r_1>0$ since $\phi \luntil_{[0,r_2]}\psi$ is equivalent to $\psi\lor(\phi \luntil_{[1,r_2]}\psi)$:

\begin{center}
  \AXC{$\phi(\x)@[t_1,t_2]$}
  \RL{\textbf{($\boxplus$)}}
  \UIC{$\boxplus_{[r_1,r_2]}\phi(\x)@[ t_1 - r_1 ,t_2 - r_2 ]$}  
  \DP
  \qquad
  \AXC{$\phi(\x)@\iota$}
    \AXC{$\psi(\y)@\iota'$}
  \RL{\textbf{($\luntil$)}}
  \BIC{$\phi(\x) \luntil_{[r_1,r_2]} \psi(\y)@(\nu-[r_1,r_2])\cap\iota$}
  \DP
\end{center}
where
$[w_1,w_2] - [r_1,r_2] := [w_1-r_2,w_2-r_1]$ and
$\nu:=(\iota+1)\cap\iota'$ (all time points where $\psi$-s are immediately preceded by $\phi$-s), and none of the involved intervals should be empty.
%
Inferences for $\boxminus$ and $\lsince$ are similar.
%
%
We denote the resulting deriver by \Rtcq. A Skolemized variant \Rtsk can be
defined similarly with temporalized versions of \gMOPO, \gCONJ, and \gEXISTS.
%
We can now lift Theorems~\ref{th:TheACzeroRewritable} and~\ref{th:TheNPDLLite} to this setting.

\begin{restatable}{theorem}{TempNPupper}
  If CQ answering in \L is UCQ-rewritable, then MTCQ answering is also UCQ-rewritable
  \blue{and $\OPtx(\L,\m)$ is in \ACzero in data complexity.}
  \blue{Moreover, $\OPtx(\DLLiteR,\m)$ is \NP-complete.}
 Let $\R\in\{\Rtcq,\Rtsk\}$.
 Then, it is \NP-complete to decide whether, given a \DLLiteR TBox $\Tmc$, a temporal ABox~$\Amc$, 
$\q(\vec{a},\iota)$ s.t.\ $\Tmc\cup\Amc\models \q(\vec{a},\iota)$, and $n$ in unary or binary
encoding,
there exists a proof in
$\R(\Tmc\cup\Amc,\q(\vec{a},\iota))$ of (tree) size at most $n$.
\end{restatable}
%



\section{Conclusion}

We started to explore a framework for proofs of answers to conjunctive queries. So far, we only considered DLs whose axioms can be formulated as existential rules, but we also want to explore more expressive logics in future work.
Other interesting research questions include derivers that combine TBox and query entailment rules, \eg \Rcq plus the rules of the ELK reasoner~\cite{DBLP:journals/jar/KazakovKS14}.
Instead of proofs, one could also try to show a canonical model to a user in order to explain query answers.
For explaining missing answers, we also want to continue investigating how to find (optimal) counter-interpretations or abduction results~\cite{DBLP:conf/ijcai/Koopmann21}.

\subsection*{Acknowledgments}
  This work was supported by DFG in grant 389792660, TRR~248
(\url{https://perspicuous-computing.science}), and QuantLA, GRK 1763
(\url{https://lat.inf.tu-dresden.de/quantla}).

\bibliographystyle{abbrv}
\bibliography{bibs}

\newpage

\appendix

\section{Additional Definitions: Hypergraphs}

\begin{definition}[Hypergraph]
A \emph{(directed, labeled) hypergraph}~\cite{DBLP:journals/cor/NielsenAP05} is a triple $H=(V,E,\el)$, where
\begin{itemize}
\item $V$ is a set of \emph{vertices},
\item $E$ is a set of \emph{hyperedges}~$(S,d)$ with a tuple of \emph{source vertices} $S$
and
\emph{target vertex} $d\in V$, and
\item $\el\colon V\to \mathcal{S}_\L$ is a \emph{labeling function} that assigns sentences to
vertices.
\end{itemize}
\end{definition}
We extend the function~$\ell$ to hyperedges as follows: $\ell(S,d):=\big((\ell(s))_{s\in S},\ell(d)\big)$.
\blue{For finite~$V$,} the \emph{size} of~$H$, denoted~$|H|$, is measured by the size of the labels of
its hyperedges: $$\card{H}:=\sum_{(S,d)\in E}\card{(S,d)}, \text{ where } \card{(S,d)}:=\card{\el(d)}+\sum_{v\in
S}\card{\el(v)}.$$
A vertex $v\in V$ is called a \emph{leaf} if it has no incoming hyperedges, \ie there is no
$(S,v)\in E$; and~$v$ is a \emph{sink} if it has no outgoing hyperedges, \ie there is no $(S,d)\in
E$ such that $v\in S$. We denote the set of all leafs and the set of all sinks in~$H$ as
$\leaf(H)$ and $\sink(H)$, respectively.

A hypergraph $H'= (V', E', \el')$ is called a \emph{subgraph} of $H= (V, E, \el)$ if $V'\subseteq V$, $E'\subseteq E$ and $\el'=\el|_{V'}$.
In this case, we also say that $H$ \emph{contains} $H'$ and write $H' \subseteq H$.
Given two hypergraphs $H_1=(V_1,E_1,\el_1)$ and $H_2=(V_2,E_2,\el_2)$ \st
$\el_1(v)=\el_2(v)$ for every $v\in V_1\cap V_2$, the \emph{union} of the two hypergraphs is
defined as $H_1 \cup H_2:=$ $(V_1\cup V_2,E_1\cup E_2, \el_1 \cup \el_2)$.

\begin{definition}[Cycle, Tree]
Given a hypergraph $H=(V,E,\el)$ and $s,t\in V$, a \emph{path} $P$ of length $q\geq 0$ in~$H$ from~$s$ to~$t$ is a sequence of vertices and
hyperedges
\[ P=(d_0,i_1,(S_1,d_1),d_1,i_2,(S_2,d_2),\dots, d_{q-1},i_q,(S_q,d_q),d_q), \]
where $d_0=s$, $d_q=t$, and $d_{j-1}$ occurs in~$S_j$ at the $i_j$-th position, for all $j$, $1\le j\le q$. If
there is such a path of length $q > 0$ in~$H$, we say that $t$ is \emph{reachable} from~$s$
in~$H$.
If $t = s$, then $P$ is called a \emph{cycle}.
The hypergraph~$H$ is \emph{acyclic} if it does not contain a cycle.
%
%
The hypergraph~$H$ is \emph{connected} if
every vertex is connected to every other vertex by a series of paths and reverse paths.

A hypergraph $H=(V,E,\el)$ is called a \emph{tree} with \emph{root} $t\in V$ if $t$ is reachable
from every vertex $v\in V\setminus\{t\}$ by exactly one path. In particular, the root is the only
sink in a tree, and all trees are acyclic and connected.
\end{definition}

\begin{definition}[Homomorphism]\label{def:homomorphism}
Let $H=(V,E,\el)$, $H'=(V',E',\el')$ be two hypergraphs.  A \emph{homomorphism} from $H$ to 
$H'$, denoted $h\colon H\rightarrow H'$, is a mapping $h\colon V\to V'$ s.t.\ for 
all~$(S,d)\in E$, one has $h(S,d):=\big((h(v))_{v\in S},\,h(d)\big)\in E'$ and, for all $v\in V$, it 
holds that~$\el'(h(v))=\el(v)$.
Such an~$h$ is an \emph{isomorphism} if it is a bijection, and its inverse, $h^-\colon H' \to H$,
is also a homomorphism.
\end{definition}

\begin{definition}[Hypergraph Unraveling]
  The \emph{unraveling} of an acyclic hypergraph $H=(V,E,\ell)$ at a vertex $v\in V$ is the tree 
  $H_T=(V_T,E_T,\ell_T)$, where $V_T$ consists of all paths in~$H$ that end in~$v$, $E_T$ 
  contains all hyperedges $((P_1,\dots,P_n),P)$ where each~$P_i$ is of the 
  form $\big(d_i,i,((d_1,\dots,d_n),d)\big)\cdot P$, and $\ell_T(P)$ is the label of the starting 
  vertex of~$P$ in~$H$.
  Moreover, the mapping $h_T\colon V_T\to V$ that maps each path to its starting vertex (and in 
  particular~$h_T(v)=v$) is a homomorphism from~$H_T$ to~$H$.
\end{definition}

The \emph{tree size} $\mtree(\p)$ of a proof~\p is defined recursively as follows:
\begin{align*}
  \mtree(v)&:=1 &&\text{for every leaf }v, \\
  \mtree(d)&:=1+\sum_{v\in S} \mtree(v), &&\text{for every }(S,d)\in E, \\
  \mtree(\p)&:=\mtree(v), &&\text{for the sink vertex }v.
\end{align*}
This recursively counts the vertices in subproofs and is equal to the size of a tree unraveling of~\p.

\section{Proof Details}

\subsection{Derivation Structures}

\LemTransformation*
\begin{proof}
Assume that we have a proof $\p$ in $\Rcq(\Tmc\cup\Amc,\q(\vec{a}))$.
We first defer all applications of schema~\EXISTS to the very end of the proof, which is possible since all other inferences remain applicable to any instance $\exists\vec{x}.\,\phi(\vec{x},\vec{a})$ of a CQ $\exists\vec{x},\vec{y}.\,\phi(\vec{x},\vec{y})$. This transformation can only decrease the size of the proof.
We then recursively change the labeling function so that it uses
Skolem terms rather than quantified variables. Specifically, starting from the leafs of~\p, we adapt
inferences of schema~\MOPO so that instead of a rule
$\psi(\vec{y},\vec{z})\to\exists\vec{u}.\,\chi(\vec{z},\vec{u})\in\Tmc$, the
corresponding Skolemized version
$\psi(\vec{y},\vec{z})\to\chi'(\vec{z})\in\Tmc^s$ is used, with existentially
quantified variables in the conclusion replaced by the corresponding Skolem
terms. The resulting hypergraph has the same size, since we only changed the
labeling, and moreover all CQs are ground.
In this process, whenever we apply a rule $\phi(\vec{x},\vec{y})\to\exists\vec{x}.\,\phi(\vec{x},\vec{y})$ generated by \TAUT to a ground CQ using \MOPO, we can omit this subproof since $\exists\vec{x}.\,\phi(\vec{x},\vec{y})$ is satisfied by the same ground atoms used to match $\phi(\vec{x},\vec{y})$.
Next, we split each modified \MOPO inference into a corresponding set of \gMOPO inferences -- this replaces each vertex~$v$ by at most $|\Tmc|$ vertices (at most for each new atom derived from the right-hand size of a rule in~\Tmc),
and thus increases the size of the hypergraph by a factor polynomial in the size of~$\p$.
Even though \Rsk contains no version of the \TAUT inference schema to generate copies of atoms, we can always use the same ground atom several times in the same inference if necessary.
At the same time, we remove all \CONJ inferences (since \gMOPO anyway uses multiple 
premises instead of a conjunction) and replace them by a single application of \gCONJ 
(conjoining all atoms relevant for~$\q(\vec{a})$).
This can also only decrease the size of the proof.
Similarly, the final \EXISTS step becomes an instance of \gEXISTS that produces the final conclusion $\q(\vec{a})$.
We obtain a proof in~$\Rsk(\Tmc^s\cup\Amc,\q(\vec{a}))$ that is of size polynomial in the size of $\p$.

Now assume that we have a proof $\p$ in $\Rsk(\Tmc^s\cup\Amc,\q(\vec{a}))$.
We transform $\p$ into a semi-linear proof in $\Rsk(\Tmc\cup\Amc,\q(\vec{a}))$.
In the beginning, we keep the Skolemization, which is eliminated in the last step.
We first collect any
ABox assertions from $\Amc$ that are used in $\p$ into a single CQ using \CONJ.
Then we aggregate inferences of \gMOPO into \MOPO-inferences. Specifically,
provided that for an $\gMOPO$ inference $(S,d)$, all labels of nodes in $S$
occur on a node $v$ that we have already aggregated, we collect all inferences
of the form $(S,d')$ (same premises, different conclusion), and transform them
into a single inference using (a Skolemized version of) \MOPO.
When $S$ contains the same atom multiple times, we first generate an appropriate number of copies using \TAUT and \MOPO; the number of such additional proof steps for each \MOPO-inference is bounded by~$|\Tmc|$ (more precisely, the maximum number of atoms on the left-hand side of any rule in~\Tmc).
%
Depending on which premises are still
needed in later inferences, we keep them in the conclusion of each \MOPO-inference or not. Since we
keep all atoms together in each step, the final application of \gCONJ is not
needed anymore. The resulting proof is de-Skolemized by replacing Skolem terms
by existentially quantified variables again using \gEXISTS. The output of this transformation
is again polynomial in the size of the original proof: the number of initial
applications of \CONJ is bounded by the size of $\p$, the aggregated \MOPO-steps can only decrease the size of the proof, and for each of these steps, we need at most $|\Tmc|$ additional \TAUT- and \MOPO-steps.
\end{proof}

\subsection{Atemporal Queries}

\TheACzero*
\begin{proof}
Fix a set $\NC^n$ of $\leq n$ individual names, and collect in
$\mathfrak{A}$ all of the (constantly many) ABoxes~$\Amc$ using only names from
$\sig(\Tmc)\cup\NC^n$ and $\a$ for which $\R(\Tmc\cup\Amc,\q(\a))$ contains
a proof of (tree) size $\leq n$; the latter could be done, \eg by simple breadth-first search for proofs up to size~$n$.
We can identify every such ABox $\Amc$ with a
CQ $q_\Amc(\a)$ obtained by replacing individual names not in $\a$ by
existentially quantified variables. We now have that, for any ABox~$\Amc$, there
exists an $\Amc'\in\mathfrak{A}$ with $\Amc\models q_{\Amc'}(\a)$ iff there
exists a proof for $\Tmc\cup\Amc\models \q(\a)$ of (tree) size at most~$n$.
Consequently, we can reduce the decision problem in the theorem to deciding
whether~$\Amc$ entails the (fixed) \emph{union of conjunctive queries} (UCQ)
\[
  \bigvee_{\Amc\in\mathfrak{A}}q_\Amc(\a).
\]
Deciding UCQ entailment without a TBox is possible in \ACzero in data complexity~\cite{AbHV-95}.
\end{proof}

\TheACzeroRewritable*
\begin{proof}
Let $Q(\vec{x})$ be a UCQ that is a rewriting of $\q(\vec{x})$ and~\Tmc, \ie such that $\Tmc\cup\Amc\models\q(\a)$ is equivalent to $\Amc\models Q(\a)$ for all ABoxes~\Amc.
For every CQ $q'(\a)\in Q(\a)$, we can
determine the minimal proof (tree) size $n_{q'}$ for $\Tmc\cup\Amc_{q'}\models
\q(\a)$, where $\Amc_{q'}$ is obtained from $q'(\a)$ by replacing every variable
by a fresh individual name.
These ABoxes represent all possibilities of an ABox entailing $\q(\a)$ w.r.t.~\Tmc (up to isomorphism), and hence they can be used as (a fixed number of) representatives in the search for small proofs of the entailment.
The computation of~$n_{q'}$ does not depend on the input data, and hence can be done offline via search in the derivation structure.
Let $n_\textit{max}$ be the maximum of the values~$n_{q'}$. To
every $n\leq n_\textit{max}$, we assign the UCQ
 \[
  Q_n=\bigvee\{q'\in Q\mid n_{q'}\leq n\}
 \]
 Given $\Amc$, $\a$, and $n$, we can now decide whether $\Tmc\cup\Amc\models
\q(\a)$ admits a proof of (tree) size $\leq n$ by 1) computing the UCQ $Q_{n'}$
for $n'=\min(n,n_\textit{max})$, and 2) checking whether $\Amc\models
Q_{n'}(\vec{a})$.
2) is the standard query answering problem which has \ACzero
data complexity. To see that the combined task 1)+2) can be done in \ACzero, we
note that $n_\textit{max}$ does not depend on the data, so that we only need to
process the least $\log_{n_\textit{max}}$ bits of~$n$ to determine~$n'$, which
can be done by a circuit of constant depth.
\end{proof}

\LemProofSize*
\begin{proof}
 \gCONJ and \gEXISTS only need to be applied once, at the very end, to produce
the query $\q(\vec{a})$. For \DLLiteR, the remaining rule \gMOPO is such that it
always has one premise that is a CQ. Consequently, we can always construct a
proof that is composed of $\lvert\q(\vec{x})\rvert$ linear proofs (one for each
atom), which are then connected using $\gCONJ$ and produce the conclusion with
\gEXISTS.
As argued in the main text, we can assume that the nesting depth
of Skolem terms is polynomially bounded by $\lvert\Tmc\rvert$, which means the
number of labels on each path is polynomially bounded by $\lvert\Tmc\rvert$ as
well.
%
%
Additionally, we can always simplify any
proof in which the same label occurs twice along some path, which means that
this polynomial bound transfers also to the depth of our proof. We obtain that
we can always find a proof of polynomial tree size.
\end{proof}

\TheNPDLLite*
\begin{proof}
  The upper bound is shown in the main text, but it still remains to show a matching
  lower bound.

 We reduce the problem of query answering to the given problem. Specifically,
let $\tup{\Tmc,\Amc}$ be a DL-Lite KB, $\q(\vec{x})$ a query and $\vec{a}$ such
that $\lvert{x}\rvert=\lvert{a}\rvert$, and suppose we want to determine whether
$\Tmc\cup\Amc\models \q(\vec{a})$. We know by Lemma~\ref{lem:proof-size}
that,
if $\Tmc\cup\Amc\models \q(\vec{a})$, then there exists a proof for this
in $\Rsk(\Tmc\cup\Amc, \q(\vec{a}))$
that
is of (tree) size at most $n=p(\lvert{\Tmc}\rvert,\lvert\q(\vec{x})\rvert)$, where $p$ is some
polynomial. By using
Lemma~\ref{lem:transformation}, we can also find such a polynomially bounded
number $n$ in case $\R=\Rcq$.

We now
construct a new KB $\Tmc_0\cup\Amc_0$ such that $\Tmc_0\cup\Amc_0\models
\q(\vec{x})$, but only with a proof of (tree) size $>n$.
$\Amc_0$ is obtained from the atoms of~$\q(\a)$ by replacing every quantified variable by a fresh individual name, and each predicate~$P$ by~$P_0$.
$\Tmc_0$ contains for
every
predicate $P$ occurring in $\q$ and every $i\in\{0,\ldots,n\}$ the CI
$P_{i}\sqsubseteq P_{i+1}$, as well as $P_n\sqsubseteq P$. Clearly,
$\Tmc_0\cup\Amc_0\models \q(\vec{a})$, and every proof for this corresponds to a
tree of depth $n+1$, and is thus larger than $n$. Moreover, $\Tmc_0$,
$\Amc_0$, and~$n$ are all of polynomial size in the size of the input to the query answering problem.
Now set $\Tmc_1=\Tmc_0\cup\Tmc$, $\Amc_1=\Amc_0\cup\Amc$. We have
$\Tmc_1\cup\Amc_1\models \q(\a)$; however, a proof of (tree) size~$\leq n$ exists in 
$\R(\Tmc_1\cup\Amc_1,\q(\a))$ iff $\Tmc\cup\Amc\models\q(\a)$.
\end{proof}

\LemPTimeCombined*
\begin{proof}
We construct a compressed version of $\Rsk(\Tmc\cup\Amc,\q(\vec{a}))$ of
polynomial size.
  We introduce for every role~$R$ the individual name $b_{\exists R}$.
 Our compressed derivation structure is defined inductively as follows, where
for a role name~$P$, we identify $P^-(a,b)$ with $P(b,a)$ and $(P^-)^-$ with~$P$.
 \begin{itemize}
  \item every axiom $\alpha\in\Tmc\cup\Amc$ has a node $v_\alpha$ with
$\el(v_\alpha)=\alpha$,
  \item for nodes $v_1$, $v_2$ with $\el(v_1)=A(a)$ and $\el(v_2)=A\sqsubseteq
B$ there is an edge $(\{v_1,v_2\},v_3)$, where $\el(v_3)=B(a)$.
  \item for nodes $v_1$, $v_2$ with $\el(v_1)=P(a,b)$ and
$\el(v_2)=P\sqsubseteq Q$, there is an edge $(\{v_1,v_2\},v_3)$ with
$\el(v_3)=Q(a,b)$,
  \item for nodes $v_1$, $v_2$ with $\el(v_1)=A(a)$ and
$\el(v_2)=A\sqsubseteq\exists P$, there is an edge $(\{v_1,v_2\},v_3)$
with $\el(v_3)=P(a,b_{\exists P^-})$,
  \item for nodes $v_1$, $v_2$ with $\el(v_1)=P(a,b)$ and
$\el(v_2)=\exists P\sqsubseteq A$, there is an edge $(\{v_1,v_2\},v_3)$
with $\el(v_3)=A(a)$,
  \item for nodes $v_1$, $v_2$ with $\el(v_1)=P(a,b)$ and
$\el(v_2)=\exists P\sqsubseteq\exists Q$, there is an edge
$(\{v_1,v_2\},v_3)$
with $\el(v_3)=Q(a,b_{\exists Q^-})$.
 \end{itemize}
 Due to the conclusions with the fresh individual names, the inferences in this
compressed derivation structure are not sound, so that it is not really a
derivation structure. But because its size is polynomial, we can use the \PTime
procedure from~\cite[Lemma~11]{ABB+-CADE21} to compute for every node $v$ a \enquote{proof} of
minimal size. To use these proofs to construct a tree proof in $\Rsk(\Tmc\cup\Amc,\q(\vec{a}))$,
we still need to match the variables in $\q(\vec{a})$ to the constants in the
derivation structure.

For every pair $\tup{t_1,t_2}$ of terms occurring in $\q(\vec{a})$ together in
an atom, and every possible assignment $(t_1\mapsto a_1, t_2\mapsto a_2)$ of
these terms to individual names from the compressed derivation structure, we
assign a cost $\gamma(t_1\mapsto a_1, t_2\mapsto a_2)$ that is the sum of the
minimal proof sizes for every atom in $\q(\vec{a})$ that contains only $t_1$
and $t_2$, with these terms replaced using the assignment. We build a labeled
graph, the \emph{cost graph}, with every node a mapping from one term in
$\q(\vec{a})$ to one constant in our compressed derivation structure, and every
edge between two nodes $(t_1\mapsto a_1)$ and $(t_2\mapsto a_2)$ labeled with
the cost $\gamma(t_1\mapsto a_1, t_2\mapsto a_2)$ (no edge if there is no edge
between $t_1$ and $t_2$ in the Gaifman graph). Every edge in the cost graph
corresponds to an edge in the Gaifman graph of $\q(\vec{a})$, but the same edge
in the Gaifman graph can be represented by several edges in the cost graph.
We can thus transform the cost graph into a directed acyclic graph, making
sure that for every edge $(t_1\mapsto a_1, t_2\mapsto a_2)$, the edge
$(t_1,t_2)$ points from the root towards the leafs of the Gaifman graph, where
we choose an arbitrary node as the root.

  We now eliminate assignments from the cost graph starting from the leafs:
 \begin{itemize}
  \item Consider a node $(t_1\mapsto a_1)$ and for some term $t_2$,
  all outgoing edges of the form $(t_1\mapsto a_1, t_2\mapsto a_2)$. Once all
the nodes $t_2\mapsto a_2$ have already been visited by the algorithm, assign to
each edge $(t_1\mapsto a_1,t_2\mapsto a_2)$ a combined cost obtained by adding
to $\gamma(t_1\mapsto a_1,t_2\mapsto a_2)$ the costs of every edge reachable
from $(t_2\mapsto a_2)$, and remove all edges for which the resulting value is
not minimal. In case several edges have a minimal value assigned, choose an
arbitrary edge to remove, so that we obtain a unique edge $(t_1\mapsto a_1,
t_2\mapsto a_2)$ for $t_1\mapsto a_1$ and $t_2$.

 \item If an assignment $(t\mapsto a)$ has no incoming edges, remove it from
the cost graph.
 \end{itemize}
 The algorithm processes each edge in the polynomially sized cost graph at most
once, and thus terminates after a polynomial number of steps with an assignment
of variables to constants appearing in the compressed derivation structure.
We can use this assignment to first construct a minimal proof for $\q(\a)$ over the
compressed derivation structure, where every atom in $\q(\a)$ has its own
independent minimal proof. Note that we cannot obtain a smaller proof in
$\q(\a)$, since every atom has a minimal proof assigned, and our elimination
procedure made sure that there cannot be a different choice of assignments of
terms in $\q(\a)$ to constants that would lead to smaller proofs anywhere else.

We then replace every constant with the
corresponding term in  $\Rsk(\Tmc\cup\Amc,\q(\a))$, starting from the original
individual names and following the structure of the proof, to obtain the desired
proof of minimal tree size.
For example, we would replace an atom $P(a,b_{\exists P^-})$ that is derived in
the proof from $A(a)$ and $A\sqsubseteq\exists P$ by $P(a,f(a))$, where $f$ is
the Skolem function for the existentially quantified variable in
$A\sqsubseteq\exists P$, and replace $b_{\exists P^-}$ by $f(a)$ also in
subsequent proof steps, provided it refers to a successor of $a$. More
generally, we replace $P(t,b_{\exists P^-})$ , where $t$ is already a Skolem
term, by $P(t,f(t))$, whenever it is derived from $A(t)$ and
$A\sqsubseteq\exists P$.
\end{proof}

We prove Theorem~\ref{the:np-hard-trees} in the following two lemmas.
\begin{restatable}{lemma}{LemNPHardModified}\label{lem:np-hard-modified}
  For tree-shaped CQs, $\OPx'(\Lmc,\mtree)$ and $\OP\sk(\Lmc,\msize)$ are \NP-hard.
\end{restatable}
\begin{proof}
 We reduce SAT to $\OPx'(\Lmc,\mtree)$. Let $c_1,\ldots,c_m$ be a set of clauses
over propositional variables $p_1,\ldots,p_n$. Each clause $c_i$ is a disjunction 
of literals of the form $p_j$ or $\overline{p_j}$, where the latter denotes the 
negation of $p_j$. In the following, we assume clauses to be represented as sets of literals.
\Wlog we assume that for every
variable $p_i$, we also have the clause $p_i\vee\overline{p_i}$. For every variable $p_i$,
we add the assertions $T(p_i)$ and $T(\overline{p}_i)$. 
 For every clause $c_i$ and every literal $l_j\in c_i$, we add an assertion
$c(c_i,l_j)$.
Furthermore, if $i<m$, we add the role
assertion $r(c_i,c_{i+1})$. The conjunctive query $\q$ is Boolean and contains
for every clause $c_i$ the atoms
 $
  c(x_i^{(c)}, x_i^{(p)}), T(x_i^{(p)})
 $
 and moreover, if $i<m$, $r(x_i^{(c)}, x_{i+1}^{(c)})$,
 so that the query is
tree-shaped. 
We have
$\Amc\models \q$ independently of whether the SAT problem has a
solution or not.

 We set our bound as $k=2+m+(m-1)+n$, which distributes as follows in a proof 
 if the set of clauses is satisfiable. Assume that $a:\{p_1,\ldots,p_n\}\rightarrow\{0,1\}$ is
 a satisfying assignment for our set of clauses. 
 \begin{itemize}
  \item 2 vertices are needed for the conclusions of \gCONJ and \gEXISTSb,
  \item $m$ vertices contain, for each clause $c_i$, the atom $c(c_i, l)$, where $l\in c_i$ 
  is made true by the assignment $a$,
  \item $m-1$ vertices contain the atom $r(c_i,c_{i+1})$ for each $i\in\{1,\ldots
m-1\}$,
  \item for each variable $p_i$, depending on whether $a(p_i)=1$ or $a(p_i)=0$, 
    we use either $T(p_i)$ or $T(\overline{p}_i)$. This needs another $n$ vertices. 
 \end{itemize}
 Note that there cannot possibly be a smaller proof, since every atom in the
query needs to be matched by some ABox assertion. Correspondingly, if the 
SAT problem has a solution, we can construct a proof of the desired tree size in $\Rsk'$, and 
if there is a proof of the desired tree size, we can extract a solution for the SAT problem 
from it. It follows that the problem is \NP-hard.
%

\blue{The same arguments apply to~\msize \wrt the original deriver \Rsk since a single atom $T(p_i)$ or $T(\overline{p}_i)$ can be used multiple times in one application of~\gEXISTS, but is only counted once.}
\end{proof}

\begin{restatable}{lemma}{LemNPHardCQStructure}\label{lem:np-hard-cq-structure}
  For tree-shaped CQs, $\OP\cq(\Lmc,\mtree)$ is \NP-hard.
\end{restatable}
%
\begin{proof}
 The central observation is that we can simulate in \Rcq the behavior of
$\gEXISTSb$ when copying atoms. Specifically, fix an inference of $\gEXISTSb$
with premise $\phi(\vec{t})$, conclusion $\exists\vec{x}.\,\phi'(\vec{x})$ and
substitution $\sigma$ s.t.\ $\exists\vec{x}.\,\phi'(\vec{x})\sigma=\phi(\vec{t})$, and
let $\alpha(\vec{x_1})$, $\alpha(\vec{x_2})\in\phi'(\vec{x})$ be such that
$\vec{x_1}\sigma=\vec{x_2}\sigma$.
To obtain a similar effect in $\Rcq$, we
would use \TAUT to derive
$\alpha(\vec{x_1})\rightarrow\exists\vec{x_1}.\,\alpha(\vec{x_1})$, which is then used
with \MOPO and the current CQ to obtain a query in which both
$\alpha(\vec{x_1})$ and $\alpha(\vec{x_2})$ occur (recall \MOPO can be used
in such a way that the new atoms are not replacing others, but are simply
added). Note that this way, we can duplicate an arbitrary number of variables
using two steps.
Let $\Amc$ and $\q$ be as in the proof of Lemma~\ref{lem:np-hard-modified}, and assume 
the SAT-problem has a satisfying assignment $a$.
We can then
construct a tree proof as follows, where this time, we collect the different atoms 
one after using $\CONJ$ into a single query. Since these inferences always have two 
premises, we only count the leaves in the following, as a binary tree with $n$ leaves and all 
other vertices binary always has $2n-1$ vertices in total. Specifically, this means that the 
number of vertices is independent on how we organize the inferences listed in the following.  
\begin{itemize}
  \item We need to collect all clauses of the form $r(c_i,c_{i+1})$, and $c(c_i,l_j)$, where 
   $l_i$ is some literal in $c_i$ evaluate to true under the chosen assignment $a$. This gives 
   $2m-1$ leaves, 
  \item Another $n$ leaves are needed to add, for each variable $p_i$, $T(p_i)$ if $a(p_i)=1$,
  and $T(\overline{p}_i)$ if $a(p_i)=0$.
%
\end{itemize}
Since this makes $2m-1+n$ leaves, we obtain that the corresponding proof must have 
$4m+2n-3$ nodes in total, independently on how we organize these inferences. We instantiate 
\TAUT with $\q(\vec{x})\rightarrow\exists\vec{x}\q(\vec{x})$, where 
$\q(\vec{x})$ is our query, containing as quantified variables exactly $\vec{x}$. 
Since the left hand side matches our ground query constructed so far, we apply it 
as final step to produce the conclusion, obtaining a tree proof of size $4+2n-1$.

Similar as before, we argue that there cannot be possibly a smaller proof for the query, since 
every leaf of this proof has to be used. Consequently, if we find a proof of tree size $4+2n-1$, we 
can construct a satisfying assignment from it, and if there is a satisfying assignment, we 
can construct a proof of size $4+2n-1$. It follows that the problem must be \NP-hard.
\end{proof}

\subsection{Temporal Queries}

\TempNPupper*
\begin{proof}
The \NP lower bound follows from the atemporal case, see Theorem~\ref{th:TheNPDLLite}. Note that the temporal inference rules do not operate on domain elements.

For the upper bound, first of all we need to show that the temporal derivation structures are sound and complete, \ie if $\Tmc\cup\Amc\models\q(\a,\iota)$ holds, then we can provide a proof for it.

The soundness of the new inferences can be easily verified from the semantics.

Completeness for both \Rtcq and \Rtsk can be shown by induction on the shape of the MTCQ $\q$. Namely, as common in the datalogMTL literature (\eg~\cite{DBLP:conf/ijcai/WalegaGKK19,DBLP:conf/ijcai/WalegaGKK20}), we organize the timeline by \emph{rulers}: we divide ABox facts into (finitely many) non-overlapping intervals whose time
points satisfy the same ABox facts. Ruler intervals can be non-regularly distributed. Since TBoxes are atemporal, inside a ruler interval all ontological inferences are the same for all timepoints.
As the basis of the induction over $\q$, we take the assumption that there are proofs for $\q(\a,\rho)=(\exists \y.\phi(\a,\y))@\rho$ if $\Tmc\cup \Amc\models \q(\a,\rho)$, for $\rho$ being a ruler interval and $\phi$ a CQ (see the correctness argument Section~\ref{sec:DS4OMQA}).
Up to this point, only need the temporal versions of \MOPO, \gMOPO, \etc The next phase of inferring temporal operators does not require \MOPO and \gMOPO anymore.
We now apply \COAL whenever possible and keep \SEP as the last step; note that \COAL aggregates formulas involving different variable names.
For temporal operators we can use similar arguments as in the case of datalogMTL over $\mathbb{R}$~\cite{BKR+-JAIR18}: our inference schemas were inspired by the chase procedure there. The cases $\land$ and $\lor$ can be shown by contradiction.

Next, for the procedure above, we can show that it builds a proof of bounded size.
\begin{lemma}\label{lem:bound-size}
There exist polynomials $p_1$, $p_2$ such that for any \DLLiteR TBox $\Tmc$, temporal ABox $\Amc$,
MTCQ $\q(\vec{x},w)$ with $\Tmc\cup\Amc\models
\q(\vec{a},\iota)$, there exists a proof
in $\Rtsk(\Tmc\cup\Amc,\q(\vec{a},\iota))$
whose tree size is bounded by $p_1(\lvert\Tmc\rvert,\lvert
\q(\vec{x})\rvert, \lvert \tem(\q) \rvert + \lvert \iota \rvert)$, if $\lvert \iota \rvert$ is finite, and by $p_2(\lvert\Tmc\rvert,\lvert
\q(\vec{x})\rvert, \lvert \Amc \rvert)$, otherwise.
\end{lemma}
\begin{proof}
The number of ruler intervals depends on the ABox: it is equal to the number of different endpoints mentioned in ABox ${}+2$ for infinities. Thus, in the case where $\iota$ is unbounded, we need to take into consideration all ruler intervals, for which we can find proofs for CQs of polynomial size according to Lemma~\ref{lem:proof-size}. 
(we only need to consider the CQs occurring in $\q$). We need at most $\lvert \q \rvert \cdot \lvert \Amc \rvert$ applications of \COAL to merge ruler intervals for all CQs (recall that \COAL takes finitely many premises); the dependence on $\lvert \Amc \rvert$ comes from the fact that a CQ may be satisfied in several non-contiguous intervals.
Each of the temporal schemas $\mathbf{(\boxplus)}$ \etc corresponds to one kind of temporal operator; again, we need to apply them to a fixed CQ (or pair of CQs) at most $\lvert \Amc \rvert$ times. We run \COAL after each, resolving one temporal operator for all intervals, \eg $\phi \luntil_{[0,3]}\psi$ over the facts $\phi@[0,5]$, $\psi@[3]$, $\psi@[6]$ is derived separately for the intervals $[0,3]$ and $[3,6]$, which we can coalesce.
The last step is \SEP if the desired certain answer interval is inside of one we obtain by the procedure above.

When $\iota=[t_1,t_2]$ is finite, we do not need to consider all ruler intervals in $\Amc$: it is enough to concentrate on the ones in $[t_1-\lvert \tem(\q)\rvert, t_2+\lvert \tem(\q)\rvert]$, an interval outside of which no facts can be part of a proof. Clearly, the size of the resulting set of relevant interval rulers is bounded by $2\cdot\lvert \tem(\q)\rvert+\lvert \iota\rvert$ (recall that $\lvert \tem(\q)\rvert$ is finite). Therefore, we can re-apply the rationale as above replacing $\lvert \Amc\rvert$ by the new bound independent of~$\Amc$.
\end{proof}

Finally, observe that the new inferences such as \COAL and $\mathbf{(\boxplus)}$ are orthogonal to the temporalized rules from Section~\ref{sec:DS4OMQA}. Thus, transitioning between derivers $\{\Rcq,\Rsk\}$ maintains the structure of temporal inferences intact and Lemma~\ref{lem:transformation} also holds in the temporal case.
Since Lemma~\ref{lem:bound-size} bounds the size of a proof, by Lemma~\ref{lem:transformation}, the \NP upper bound follows.

We can prove the rewritability claim by induction on the structure of the MTCQ.
CQs are UCQ-rewritable by assumption.
For the temporal dimension, over \Zbb any MTCQ with bounded intervals can be transformed into an equivalent formula with only \emph{next} temporal operator ($\lnext$ in LTL notation), the semantics of which is the same as of punctual temporal operators $\boxminus_{[1,1]}$ and $\boxplus_{[1,1]}$ in Figure~\ref{fig:tcq-semantics}. For example, by the recursive rewriting procedures below, in finitely many steps we reach a positive formula without $\boxplus$ and $\luntil$:
$$\boxplus_{[r_1,r_2]}\phi \equiv 
\begin{cases}
 \lnext(\boxplus_{[r_1-1,r_2-1]}\phi) , \text{ if } r_1 > 0, \\
 \phi \land \lnext(\boxplus_{[0,r_2-1]} \phi), \text{ if } r_1 = 0 \text{ and } r_2 > 0,\\
  \phi, \text{ otherwise.}
\end{cases}
$$
\noindent and 
 $$
 \phi \luntil_{[r_1,r_2]} \psi \equiv 
\begin{cases} 
 \phi \land \lnext(\phi\luntil_{[r_1-1,r_2-1]} \psi), \text{ if } r_1 > 0, \\
 \psi \lor (\phi \land \lnext(\phi\luntil_{[0,r_2-1]} \psi)), \text{ if } r_1 = 0 \text{ and } r_2 > 0,\\
  \psi, \text{ otherwise.}
 \end{cases}
 $$ 
A celebrated result of Kamp~\cite{phd-kamp} is that, over both $(\Zbb, <)$ and $(\mathbb{R}, <)$, all LTL operators are definable in $\FO(<)$, first-order logic with the built-in linear order. Thus, any MTCQ is $\FO(<)$-rewritable. In terms of circuit complexity, this means that MTCQ answering is in LogTime-uniform \ACzero for data complexity~\cite{Immerman99}. 


We can see that our \FO-rewriting is a positive-existential (PE) first-order formula, \ie uses only $\land$, $\lor$, and existential quantification with the linear order seen as a predicate. Such a formula can be directly converted into an equivalent UCQ~\cite{AbHV-95}.
The second statement of the proof follows from as similar argument as in the proof of Theorem~\ref{th:TheACzeroRewritable}, \ie we can pre-compute the minimal proof (tree) size for each of the CQs in this (fixed) UCQ.
\end{proof}

\end{document}